\documentclass[11pt]{article}

\newif\ifnotes

\usepackage{commands}
\usepackage[margin=1in]{geometry}
\usepackage{amssymb}
\usepackage{amsmath}
\usepackage{amsthm}
\usepackage{amsfonts}
\usepackage{bm}
\usepackage{enumitem}
\usepackage{color}
\usepackage{comment}
\usepackage[capitalize]{cleveref}
\usepackage[dvipsnames]{xcolor}
\usepackage{float}
\usepackage[T1]{fontenc}
\usepackage{todonotes}
\usepackage{asymptote}
\usepackage{mdframed}
\usepackage[most]{tcolorbox}
\usepackage{enumitem}
\usepackage{framed}
\usepackage{mdframed}
\usepackage{scrextend}
\usepackage{bbm}
\usepackage{thm-restate}
\usepackage{xcolor}
\usepackage{algorithm}
\usepackage[commentColor=CarnationPink]{algpseudocodex}
\usepackage{float}
\usepackage{tikz}
\usepackage{mathtools}
\usetikzlibrary{positioning, fit, decorations.pathreplacing, arrows.meta, calc}

\tikzset{
    ncbar angle/.initial=90,
    ncbar/.style={
        to path=(\tikztostart)
        -- ($(\tikztostart)!#1!\pgfkeysvalueof{/tikz/ncbar angle}:(\tikztotarget)$)
        -- ($(\tikztotarget)!($(\tikztostart)!#1!\pgfkeysvalueof{/tikz/ncbar angle}:(\tikztotarget)$)!\pgfkeysvalueof{/tikz/ncbar angle}:(\tikztostart)$)
        -- (\tikztotarget)
    },
    ncbar/.default=0.5cm,
}

\tikzset{square left brace/.style={ncbar=2mm}}
\tikzset{square right brace/.style={ncbar=-2mm}}

\newcommand{\vgnote}[1]{$\ll$\textsf{\color{blue} Venkat: { #1}}$\gg$}

\newcommand{\mgnote}[1]{\ifnotes $\ll$\textsf{\color{magenta} Meghal: { #1}}$\gg$ \fi}
\notestrue

\usepackage[T1]{fontenc}

\usepackage{hyperref}
\hypersetup{
    colorlinks=true,
    linkcolor=blue,
    filecolor=blue,
    citecolor=magenta,
    urlcolor=magenta,
}
\usepackage[hyperpageref]{backref}

\Crefname{theorem}{Theorem}{Theorems}
\Crefname{claim}{Claim}{Claims}
\Crefname{lemma}{Lemma}{Lemmas}
\Crefname{proposition}{Proposition}{Propositions}
\Crefname{corollary}{Corollary}{Corollaries}
\Crefname{definition}{Definition}{Definitions}

\newcommand{\ECC}{\mathsf{ECC}}
\newcommand{\DEC}{\mathsf{DEC}}

\newcommand{\qlist}{\mathsf{qlist}}

\newcommand{\enc}{\mathsf{enc}}
\newcommand{\dec}{\mathsf{dec}}

\newcommand{\LDC}{\mathsf{LDC}}

\newcommand{\C}{\mathsf{C}}

\newcommand{\bbN}{\mathbb{N}}

\newcommand{\bbF}{\mathbb{F}}

\newcommand{\bbK}{\mathbb{K}}
\newcommand{\bbE}{\mathbb{E}}
\newcommand{\cA}{\mathcal{A}}

\newcommand{\cD}{\mathcal{D}}

\newcommand{\cF}{\mathcal{F}}

\newcommand{\cM}{\mathcal{M}}

\makeatletter
\newcommand{\customlabel}[2]{%
   \protected@write \@auxout {}{\string \newlabel {#1}{{#2}{\thepage}{#2}{#1}{}} }%
   \hypertarget{#1}{#2}
}
\makeatother

\newcounter{casenum}

\newcounter{subcasenum}

\newcounter{casenump}

\newcommand{\casep}[2]{
    \ifthenelse{\equal{\value{casenump}}{0}}{
    \vskip.5\baselineskip\par\noindent
    }{}
    {\it Case \arabic{casenump}:} {\it #1}
    \vskip0.1\baselineskip
    \begin{addmargin}[1.5em]{1em}
    #2
    \end{addmargin}
    \addtocounter{casenump}{1}
}

\newcounter{subcasenump}

\begin{document}

\title{Tight bounds for stream decodable error-correcting codes}
\author{
    Meghal Gupta\thanks{E-mail: \texttt{meghal@berkeley.edu}. This author was supported by an NSF GRFP Fellowship and Venkat Guruswami's Simons Investigator award.}\\UC Berkeley \and 
    Venkatesan Guruswami\thanks{E-mail: \texttt{venkatg@berkeley.edu}. Supported in part by a Simons Investigator award and NSF grants 2210823 and CCF-2211972.}\\UC Berkeley \and 
    Mihir Singhal\thanks{E-mail: \texttt{mihirs@berkeley.edu}. This author was supported by an NSF GRFP Fellowship and NSF CCF-2210823.}\\UC Berkeley
}
\date{\today}

\sloppy
\maketitle
\begin{abstract}
In order to communicate a message over a noisy channel, a sender (Alice) uses an error-correcting code to encode her message $x$ into a codeword. The receiver (Bob) decodes it correctly whenever there is at most a small constant fraction of adversarial error in the transmitted codeword. This work investigates the setting where Bob is computationally bounded. Specifically, Bob receives the message as a \emph{stream} and must process it and write $x$ in order to a write-only tape while using low (say polylogarithmic) space. We show three basic results about this setting, which are informally as follows:
%(please refer to the paper for the precise and general statements of our results):
\begin{enumerate}[label=(\arabic*)]
    \item There is a stream decodable code of near-quadratic length.
    \item There is no stream decodable code of sub-quadratic length.
    \item If Bob need only compute a private linear function of the input bits, instead of writing them all to the output tape, there is a stream decodable code of near-linear length.
\end{enumerate}

%\vgnote{Made some small edits. Might be good to add a sentence each about the techniques for the results, but perhaps it's okay to keep it short.} \mgnote{looks great! personally, I like short abstracts anyway}
\end{abstract}
\thispagestyle{empty}
\newpage
\tableofcontents
\pagenumbering{roman}
\newpage
\pagenumbering{arabic}

\section{Introduction}

Consider the following task: a sender wishes to communicate a message to a receiver that it receives and processes bit-by-bit. This scenario arises, for instance, in automatic control applications, where a device receives an incoming stream of instructions that it executes in sequence. Concretely, consider a small satellite receiving instructions from a large control center on the ground. The control center wants to send the satellite instructions in a way that satisfies two properties:
\vspace{-0.5em}
\begin{itemize}
\itemsep=0ex
    \item \textbf{Error resilience.} The satellite should execute the correct set of instructions even if a constant fraction of the transmission is corrupted.
    \item \textbf{Low-memory decoding.} The satellite should be able to process the instructions in order while only using limited space (significantly less than the total length of the instructions). 
\end{itemize}

Sending the list of instructions $x_1\ldots x_n$ directly, although easy to process and execute one-by-one, is not resilient to error, and thus an unsatisfactory solution. On the other hand, encoding $x_1\ldots x_n$ into $\ECC(x_1\ldots x_n)$ with a standard error-correcting code~\cite{Shannon48, Hamming50} is resilient to error, but requires the receiver to store the whole stream to decode, which is too much space. An intermediate approach would be to encode the individual instructions each by error-correcting codes as $\ECC(x_1)\ECC(x_2)\ldots\ECC(x_n)$. However, this does not withstand a constant overall fraction of corruption: the adversary can corrupt $\ECC(x_1)$ entirely, using only a $1/n$ fraction of corruption, and thus never allow the satellite to recover $x_1$. What we would like is a code that encodes the message globally but can be decoded in low space as the encoded message arrives.
\vspace{-0.5em}
\paragraph{Stream-decodable codes.} This motivates the notion of a \emph{stream-decodable} error-correcting code. In this model, we require that the receiver can output the entire message $x_1\ldots x_n$ when any small fraction $\rho$ of the message is adversarially corrupted while using low space (for example, $\polylog(n)$ space) to process the transmission bit-by-bit. 
More formally,\footnote{The technical definition of the model is in Section~\ref{sec:model}.} a stream decodable code has the following two components:
\vspace{-0.5em}
\begin{itemize}
\itemsep=0ex
    \item An encoding function $\enc: \{0,1\}^n \to \{0,1\}^{m(n)}$.
    \item A randomized decoding algorithm $\dec: \{0,1\}^{m(n)} \to \{0,1\}^n$ that uses $s(n)$ space ($s(n)$ is much smaller than $n$: for instance, $s(n)=\polylog(n)$). For all $x$, and for any $z\in \{0,1\}^{m(n)}$ within Hamming distance $\rho m(n)$ of $\enc(x)$, it should hold that $\dec(z)$ outputs $x$ with high probability. 
    %\vgnote{Allow randomness and small failure probability.} 
\end{itemize}

It is not clear that such codes should exist at all for any $s(n)=o(n)$, even with any positive $\rho$ error and any communication blow-up $m(n)$. In particular, a standard error-correcting code could require storing the full encoding at once to process, and so it requires $s(n)=n$. 

Our first result constructs stream-decodable codes that achieve the following parameters (see Theorem~\ref{thm:n2const} for a precise statement).
\vspace{-0.5em}
\begin{itemize}
\itemsep=0ex
\item Error resilience of $\rho=\frac14-\eps$ for any $\eps>0$, matching the best possible error resilience of standard error-correcting codes.
    \item Near-quadratic blow-up in communication: $m(n)=\frac{n^{2+r(s(n))}}{s(n)}$ (here $r(s(n))$ is small -- typically $o(1)$, but when $s(n)=\log(n)^t$, then $r(s(n))=1/t$). This is a larger blow-up in communication than incurred by standard error-correcting codes. %\mgnote{It would be good to choose a word for our version of near quadratic/linear and use it consistently.} \vgnote{I suggest near-quadratic}
\end{itemize} 
The construction itself is quite simple: it encodes the message by a locally decodable code of near-linear length, and repeats the encoding $O(n)$ times. The more interesting part is the corresponding decoding algorithm for Bob. To this end, we work with stronger local decoding guarantees, specifically having access to soft information for unique decoding, and local list decoding with advice from close to $1/2$ errors. We provide an overview of the approach in \cref{subsec:overview-quad-const}. 
%\vgmargincomment{Added couple of sentences on techniques.}

\medskip\noindent\textbf{A matching lower bound.} Our next result demonstrates, surprisingly, that the communication blowup of our codes is essentially optimal: any stream-decodable code requires transmission length $m(n)=\Omega\left(\frac{n^2}{s(n)}\right)$, in contrast to the standard error-correction model. (See Theorem~\ref{thm:n2lower} for the precise statetemt.)  This result is surprising and notable because it obtains a strong lower bound on an information-theoretic quantity (the codeword blow-up) leveraging the computational restriction of space-boundedness. 
The lower bound is established by carefully controlling the set of message indices that the decoder can output when processing successive blocks of sub-linear size of the stream. We provide a high level overview of this approach in Section~\ref{subsec:overview-lb}.
%\mgnote{It would be good if we can say some buzzwords about our techniques: Our lower bound is via... usually people care about techniques with streaming lower bounds idk.}
%\mgnote{In general, I feel I focused a bit too much in the intro on the positive construction aspect of our results, whereas in reality the code is not that good and actually the lower bound/tightness is probably the more interesting part. If we can emphasize that aspect, it would be good.}\vgnote{I agree; I put in a couple of sentences}\mgnote{even just drawing attention to the fact that the lower bound is interesting because info theoretical n is enough but specifically the space restriction makes it need to be more, idk just some words drawing attention to the lower bound being a main contribution}

\medskip\noindent\textbf{Comparison to \cite{GuptaZ23}.} Our notion of stream-decodable codes is similar to the model recently introduced by Gupta and Zhang~\cite{GuptaZ23}. Their model is identical to ours, except that instead of outputting $x_1\ldots x_n$, the decoder need only output a single bit $f(x_1\ldots x_n)$. Here, the function $f$ represents the output of an arbitrary streaming algorithm performed on $x_1\ldots x_n$, so it must be possible to compute in $s(n)$ space in a streaming manner. The function $f$ is unknown to Alice (or else she could simply send the value $f(x_1\ldots x_n)$ to Bob) but known to the adversary causing the errors. One could imagine, for example, that $f$ is an arbitrary linear function of $x_1\ldots x_n$, or one's physical location after executing some (possibly non-commutative) sequence of movements $x_1\ldots x_n$.

For this problem, \cite{GuptaZ23} provide a scheme requiring slightly larger than $O(n^4)$ encoding length.\footnote{Specifically, if $s(n)=\log(n)^t$, their code requires $n^{4+O(1/t)}$ space.} Our notion of a stream decodable code is necessarily stronger: that is, if the decoder can write $x_1\ldots x_n$ to an output tape in that order, they can also compute the output of any streaming algorithm $f$ in low space. As such, our construction improves upon theirs, providing a nearly quadratic-length code for their problem. 

Furthermore,~\cite{GuptaZ23} specifically investigate the scenario where Alice knows beforehand that Bob's function $f$ is a linear function of $x_1\ldots x_n$. 
%Note that the case where Bob wants to just learn a single bit $x_i$ (for $i$ unknown to Alice) is already captured by this case. Furthermore, computing linear functions of data is an important goal in itself.
For this restricted class of functions, Gupta and Zhang demonstrate a scheme that uses slightly larger than $O(n^2)$ encoding length. 

\medskip\noindent\textbf{Near-linear length code for stream computation of linear functions.} 
Our final result, stated precisely as \cref{thm:linear-const}, is a new scheme for stream-decoding linear function. Specifically, we improve upon Gupta and Zhang's result, demonstrating a scheme that requires near-linear communication in $n$ for computing linear functions. This is achieved using a tensor product of locally decodable codes for the encoding, and a careful recursive decoding approach to recover the desired linear function of the message; see \cref{subsec:overview-linear-func} for an overview.

\subsection{The model definition} \label{sec:model}

Before we provide the technical statements of our three main results, let us formally define the model of stream decodable codes. A \emph{$(\rho, m(n), s(n))$-stream coding scheme} with probability $p$ of success consists of the following:
\vspace{-0.5em}
\begin{itemize}
\itemsep=0ex
    \item An explicit family of encoding algorithms $\enc = \{ \enc_n : \{ 0, 1 \}^n \rightarrow \{ 0, 1 \}^{m(n)} \}$ with encoding time $\poly(n)$.
    \item An explicit family of randomized decoding algorithms $\dec = \{ \dec_n : \{ 0, 1 \}^{m(n)} \rightarrow \{ 0, 1 \}^n \}$ that read the input in a stream and are permitted $s(n)$ memory $\poly(n)$ time. The output is written to a write-only tape that writes only left-to-right. Whenever the Hamming distance $\Delta(z, \enc(x)) < \rho m$, then $\dec(z)$ outputs $x$ except with probability $p$. 
    %\msmargincomment{are we going to make a distinction between erasure and deletion?} \mgmargincomment{not sure. it might make the statements clunkier, and then we would maybe want to do theorem 1.1/1.3 for 1/2 error resilience in the erasure world to paint the full picture. otoh we know 1.2 holds for erasures so idk?}
\end{itemize}

%Informally, the scheme encodes a message of length $n$ into a message of length $m(n)$, and can be decoded by a streaming algorithm with memory $s(n)$ whenever at most $\rho \cdot m(n)$ bits of the code are adversarially corrupted. 

More generally, a $(\rho, m(n), s(n))$-stream coding scheme for a family of functions $\cF = \{f : \{0,1\}^* \to \{0,1\}^* \}$ consists of the following similar components, with the same time and space guarantees as above: 
\vspace{-0.5em}
\begin{itemize}
\itemsep=0ex
    \item An explicit family of encoding algorithms $\enc^{(\cF)} = \{ \enc^{(\cF)}_n : \{ 0, 1 \}^n \rightarrow \{ 0, 1 \}^{m(n)} \}$.
    \item For each $f\in \cF$, an explicit family of randomized decoding algorithms $\dec^{(f)} = \{ \dec^{(f)}_n : \{ 0, 1 \}^{m(n)} \rightarrow \{ 0, 1 \}^n \}$ that read the input in a stream and write to a left-to-right output tape. Whenever the Hamming distance $\Delta(z, \enc(x)) < \rho m$, then $\dec^{(f)}(z)$ outputs $f(x)$ except with probability $p$. 
    %\msmargincomment{are we going to make a distinction between erasure and deletion?} \mgmargincomment{not sure. it might make the statements clunkier, and then we would maybe want to do theorem 1.1/1.3 for 1/2 error resilience in the erasure world to paint the full picture. otoh we know 1.2 holds for erasures so idk?}
\end{itemize}
We emphasize that the encoding has no knowledge of $f$, only of the family $\cF$, while the decoding algorithm must succeed for all $f$.

\subsection{Our results}

In this section, we formally state our results. In this section, when we use the phrase ``absolute constant'' to describe any parameter, we mean that any asymptotic notation henceforth may have constants depending on that absolute constant. 

The first result is a stream decodable error-correcting code incurring approximately quadratic blow-up in communication. 
% We state the precise bounds for the regimes where the streaming decoding algorithm has space $s(n)=n^\delta$ and $s(n)=(\log n)^{1/\delta}$, although we emphasize that in both regimes the code length is approximately $n^2$.

% \begin{restatable}{theorem}{quadconst}
% \label{thm:n2const}
%     Let $\eps,\delta>0$ be sufficiently small parameters. 
%     \begin{enumerate} [label=(\arabic*)]
%         \item There is a $\left(\frac14-\eps, n^{2-\delta+o_{\eps,\delta}(1)}, n^\delta \right)$-coding scheme for streams.
%         \item There is a $\left(\frac14-\eps, n^{2+O_\eps(\delta)}, (\log n)^{1/\delta} \right)$-coding scheme for streams.
%     \end{enumerate}
% \end{restatable}
% \mgnote{I don't really know how to reasonably state this bound for general $s(n)$.}

\begin{restatable}{theorem}{quadconst}
\label{thm:n2const}
    Fix an absolute constant $\eps>0$. Then, for some large absolute constants $C=C(\eps)$ and $c=c(\eps)$, the following hold.
    %Let $r(n)=O\left(\frac{\log\log n }{\log s(n)} \right)+o(1)$. There is a $\left(\frac14-\eps, \frac{n^{2+r(n)}}{s(n)}, s(n)  \right)$-coding scheme for streams.
    \vspace{-0.5em}
    \begin{itemize}
    \itemsep=0ex
        \item If $s(n)=(\log n)^t$ for some absolute constant $t>C$, then there is a $\left(\frac14-\eps, n^{2+c/t}, s(n) \right)$-stream coding scheme. 
        \item For any function $s(n) = (\log n)^{\omega(1)}$, there is a $\left(\frac14-\eps, \frac{n^{2+o(1)}}{s(n)}, s(n) \right)$-stream coding scheme. (Here, the implicit constants in the $o(1)$ may depend on those in the $\omega(1)$.)
    \end{itemize}
    \vspace{-0.5em}
    Both schemes succeed with probability $1-\frac{1}{n^{\omega(1)}}$.
\end{restatable}

The second result establishes that Theorem~\ref{thm:n2const} is essentially optimal. That is, any encoding of a message that permits a low-space streaming algorithm to decode requires $\Omega\left( \frac{n^2}{s(n)} \right)$ length.

\begin{restatable}{theorem}{quadlower}
\label{thm:n2lower}
    Fix an absolute constant $\rho>0$ and let the space for the decoding algorithm be $s(n) \ge \log n$. Suppose there is a $\left(\rho, m, s(n) \right)$-coding scheme for streams that succeeds with probability at least $1-\frac{1}{2n^2}$. Then, $m = \Omega\paf{n^2}{s(n)}$.
\end{restatable} 
%\vgmargincomment{theorem statement should say something about failure probability}

The final result states that the encoding length can be made almost linear for stream coding schemes that compute a linear function $f(x)$. Here, the decoder need only output a private linear function $f$ of the input bits rather than the entire input. 
When $s(n)=n^\delta$ for sufficiently small $n$, this can even be made exactly linear, which is optimal.
%We state the precise bounds for the regimes where the streaming decoding algorithm has space $s(n)=n^\delta$ and $s(n)=(\log n)^{1/\delta}$.
%As with Theorem~\ref{thm:n2const}, we state the precise bounds for the regimes where $s(n)=n^\delta$ and $s(n)=(\log n)^{1/\delta}$.

% \begin{restatable}{theorem}{linearconst}
% \label{thm:linear-const}
%     Fix sufficiently small parameters $\eps,\delta>0$. 
%     \begin{enumerate} [label=(\arabic*)]
%         \item There is a $\left(\frac14-\eps, O(n), n^\delta \right)$-coding scheme for the family of linear functions.
%         \item There is a $\left(\frac14-\eps, O(n^{1+o(1)}), (\log n)^{1/\delta} \right)$-coding scheme for the family of linear functions.
%     \end{enumerate}
% \end{restatable}

\begin{restatable}{theorem}{linearconst}
\label{thm:linear-const}
    Fix an absolute constant $\eps>0$. Then, for some large absolute constants $C=C(\eps)$ and $c=c(\eps)$, the following hold.
    \vspace{-0.5em}
    \begin{itemize}
    \itemsep=0ex
        \item If $s(n)=(\log n)^t$ for some absolute constant $t>C$, then there is a $\left(\frac14-\eps, n^{1+c/\sqrt{t}}, s(n) \right)$-stream coding scheme for the family of linear functions.
        %\item If $s(n)=(\log n)^{\omega(1)}$, then there is a $\left(\frac14-\eps, O(n^{1+o(1)}), s(n) \right)$-coding scheme for the family of linear functions. 
        \item If $s(n)=\Omega(n^\delta)$ for some absolute constant $\delta>0$, then there is a $\left(\frac14-\eps, O(n), s(n) \right)$-coding scheme for the family of linear functions.
    \end{itemize}
    \vspace{-0.5em}
    Both schemes succeed with probability $1-\frac{1}{n^{\omega(1)}}$.
    %For any fixed $\delta$, there is a $\left(\frac14-\eps, O(n^{1+\delta}), s(n) \right)$-coding scheme for the family of linear functions. 
    %In the specific case where $s(n)=\Omega(n^\gamma)$ for some fixed $\gamma>0$, there is a $\left(\frac14-\eps, O(n), s(n) \right)$-coding scheme for the family of linear functions. Both schemes succeed with probability $1-\frac{1}{n^{\omega(1)}}$.
    %\vgnote{Should $O(n)$ be $O_{\eps,\delta}(n)$?} \mgnote{If we say $\eps$ and $\delta$ are fixed beforehand, is that still needed?}
\end{restatable}

% We remark that when $s(n)=\omega\left((\log n)^t\right)$, this provides a $\left(\frac14-\eps, n^{1+O(1/\sqrt{t})}, s(n) \right)$-stream coding scheme.

\subsection{Discussion and further directions}\label{sec:open}

\paragraph{Tightening lower order terms.} Both Theorem~\ref{thm:n2const} and Theorem~\ref{thm:linear-const} construct codes that are not optimal in lower-order terms. It may be possible to construct stream coding schemes of exactly length $O\left(\frac{n^2}{s(n)}\right)$ and stream coding schemes for linear functions of length $O(n)$. This is an interesting direction for future work.

Specifically, in the case of linear functions, it may be possible to construct constant rate codes. Interestingly, we can pose this question for an even more restrictive class of functions than linear functions: the class of index functions. These are the functions $f_i(x)=x_i$ for all $i$. We do not even know if constant rate stream coding schemes exist here, when $s(n)$ is sufficiently small, say $\polylog(n)$.

One simple strategy is to encode with a locally decodable code requiring $Q=s(n)^{O(1)}$ queries to recover an index. The decoder can then ignore all the indices except the $Q$ they need to recover their target index, and in $\poly(Q)=s(n)$ space, recover any individual index of the message. Indeed, for our constructions in both Theorem~\ref{thm:n2const} and Theorem~\ref{thm:linear-const}, this procedure is an important primitive. Unfortunately, the best locally decodable codes for polylogarithmic locality have super-linear encoding length~\cite{Yekhanin11}, and so are not constant rate.

However, this is not necessarily the only way. Barring a constant rate construction of $\polylog(n)$-query locally decodable codes, can we construct constant rate stream decoding schemes? We remark that if one removes the streaming requirement and only requires that the decoder be low space with arbitrary queries, constant rate codes are known~\cite{Spielman96,GuruswamiK08}.
\vspace{-2ex}
\paragraph{Improvements to the lower bound.} We will discuss a few potential strengthenings to Theorem~\ref{thm:n2lower}.

The work of~\cite{GuptaZ23} initially proposes the model of stream coding schemes where the decoder need only output $f(x_1\ldots x_n)$, for an arbitrary choice of Boolean function $f$ that can be computed by receiving $x_1\ldots x_n$ in a stream in $s(n)$ space. The simplest way to accomplish this task is to compute $x_1\ldots x_n$ in order and perform the streaming computation of $f$ as each bit is discovered. Our lower bound shows that this method requires encoding length $\Omega\left( \frac{n^2}{s(n)} \right)$, but there could be a different way. Nonetheless, we conjecture that the lower bound of $\Omega\left( \frac{n^2}{s(n)} \right)$ encoding length holds for any stream coding scheme for the class of all Boolean functions computed by $s(n)$-space streaming algorithms. 
%\mgnote{We could elaborate on  potential constructions for this.} \vgnote{Can we conjecture a specific candidate hard function that's computable in polylog streaming space? I recall some non-abelian group product related function that we discussed.}

Secondly, our lower bound in Theorem~\ref{thm:n2lower} only disproves stream coding schemes where the decoder outputs $x$ with probability $1-\frac{1}{\poly(n)}$. However, random guessing only outputs $x$ correctly with probability $\frac{1}{2^n}$. We conjecture that a stream coding scheme requires $\Omega\left( \frac{n^2}{s(n)} \right)$ space to output $x$ even if we only require the success probability to be $\frac{1}{\poly(n)}$.

\subsection{Related Work}

Aside from the connection to~\cite{GuptaZ23}, we discuss the relation of our work different models of error-correcting codes and to streaming algorithms.

\medskip\noindent\textbf{Efficiency of error-correction algorithms.}
Our work explores the model of stream decodable error-correcting codes where the decoder must be low-space and read the codeword bits in a single pass. Without the latter restriction, it is known how to construct asymptotically good codes for which a \emph{logspace} decoder can correct a constant fraction of errors, and in fact one can also have a logspace encoder~\cite{Spielman96,GuruswamiK08}. Note that the decoder is not restricted to a single pass on the codeword. Since one typically receives a communicated codeword as a stream, a natural question is whether such results extends to low space decoding in a stream.
This is our focus, and we show that for error-correction in this setting, one cannot have codes of constant rate, and in fact a near quadratic blow-up in message length becomes necessary. 
%\vgnote{Changed here and added citations. Hope last sentence is okay since the lower bound is only if we insist on in-order decoding}
%Error correction~\cite{Shannon48,Hamming50}, is one of the most well-studied problems in computer science. 
%A variety of works have contributed to efficient algorithms for decoding, both in time and space complexity. 
%\mgnote{Venkat, can you put in a sentence and some citations for efficient time decodable ecc's. I also recall you pointing out a result that low space decoding is possible. can you link that as well?}. \cite{} shows that low space (polylogarithmic space) decoding of error-correcting codes is possible. 

\medskip\noindent
\textbf{Codes against streaming channels.} 
Streaming in the context of error-correction has previously been considered for \emph{channels} which are low-space (say logarithmic in the code length) and cause errors in an online fashion as they read the input codeword in a single pass. This was part of a more general line of work on coding against computationally bounded channels whose systematic study was initiated in \cite{GuruswamiS16}. List-decodable codes with optimal rate against such channels were constructed in \cite{ShaltielS21a} for all error fractions $p \in [0,1/2)$, and their decoding time improved in \cite{KoppartySS19}. More recently and surprisingly, even unique-decodable codes with optimal rate (for at most  fractions $p <1/4$ of errors caused by a streaming channel) were constructed in \cite{ShaltielS21b}. 

There is also beautiful work on causal channels, which must read the codeword in one pass, but there is no space restriction~\cite{CJL-stoc15}. In contrast, our work is in the model where the \emph{receiver} is the computationally bounded party. 

Additionally, the authors of \cite{Franklin15} consider a related version of the problem we consider, where the encoder also receives the message as a stream rather than all at once, but the encoder and decoder are permitted shared randomness. In this setting, they show that it is possible to achieve a constant-rate encoding with any constant fraction of errors less than 1.
%\vgnote{Added clarifications}

\medskip\noindent\textbf{Locally decodable codes.} One specific type of error-correcting codes related to our result is that of \emph{locally decodable codes}. Locally decodable codes \cite{Yekhanin12} can be viewed as a low-time and low-space version of error-correcting codes, where the goal is to learn a single bit of the original message. In constrast, for us, the decoder must be able to piece the entire stream with the tradeoff that the decoder accesses the entire encoding via a stream rather than via query access. Locally decodable codes have been constructed in a variety of different parameter regimes, including constant query~\cite{Efremenko09, DvirGY11} and rates approaching $1$~\cite{KoppartySY14}. In our work, we will use Reed-Muller codes~\cite{Muller54, Reed54} that achieve $\polylog(n)$ query complexity and slightly super-linear block length.

As discussed in Section~\ref{sec:open}, our work also connects to locally decodable codes as a relaxation of the query model. Specifically, $q$-query locally decodable codes are $\poly(q)$ space stream coding schemes for the family of linear functions (as long as the locally decodable code permits $\poly(q)$ decoding space). Thus, our model can be viewed as a simpler setting than local decoding in which to construct high rate codes. In particular, the existence of constant rate stream decodable codes for index functions may be easier to resolve than constant rate polylogarithmic-query locally decodable codes.

\medskip\noindent\textbf{Streaming Algorithms.}
The algorithms in this work are \emph{streaming algorithms} for processing noisy encoded communication.

Streaming algorithms are a prolific field of research with algorithms for a multitude of problems, including approximate counting~\cite{Morris78} on approximate counting, heavy hitters~\cite{CharikarCF02}, $\ell_p$ approximation~\cite{AlonMS96, MonemizadehW10, IndykW05}, and identifying a nonzero entry in a vector (for turnstile algorithms)~\cite{MonemizadehW10}. Many works, for example \cite{Garg21, Chen16, Ma21, BenJWY22}, also consider the problem of processing \emph{noisy} data using streaming algorithms. \cite{Garg21} shows a memory lower bound for learning a parity function with noisy samples of random linear equations.

However, the typical streaming setting is quite different from our setting. The algorithms mentioned above are used to process adversarial or ``random'' data. Our algorithms on the other hand process carefully formatted streams in the presence of communication noise, rather than sample or data noise. Our streaming algorithms are for processing formatted communication rather than processing data.
\section{Overview of techniques} \label{sec:overview}

\subsection{Stream-decodable codes of near-quadratic length}
\label{subsec:overview-quad-const}
We start by describing our construction of \cref{thm:n2const}, in the case where $s = (\log n)^t$ for sufficiently large $t$ (the other case is very similar).

\medskip\noindent\textbf{Achieving constant error rate.}
First, we will describe how to achieve a code which is resilient against a sufficiently small constant error rate (say $0.01$). Take a Reed-Muller code $\LDC: \{0, 1\}^n \to \{0, 1\}^N$ with appropriate parameters which is locally correctable and locally decodable up to constant fraction of errors (say a $0.2$ fraction). Alice's encoding is just $y=\LDC(x)$ repeated $10n$ times, for a total length of $10nN$. This is the same encoding used by~\cite{GuptaZ23}.

Next, we describe Bob's decoding algorithm. Bob's goal is to output one more bit $x_i$ of the stream after each received block (copy of $y$) where there is less than $0.1$ corruption. Doing this requires two things: (1) that he can decode any single bit of $x$ from a block with $<0.2$ corruptions and (2) that he can detect whether a block had sufficiently little (around $0.15$ or less) fraction of corruptions. The first goal is simple: while receiving a copy of $y$, Bob uses local decoding to figure out the next bit of $x$ and will be correct with high probability if the corruption was lower than $0.2$. Then, he outputs his guess if the block passes the check for (2).

Now, we explain how Bob achieves (2). Essentially, Bob will use the first $5n$ (corrupted) blocks to find an ``estimate'' of the true value of $y$ (not doing (1) at all during this phase), and for each of the second $5n$ blocks, checks if this estimate matches the received block well enough (indicating low corruption), and if so, he does (1). More precisely, he picks $v = (\log n)^2$ uniformly random indices $j_1, \dots, j_v \in [N]$. Then, for the first $5n$ (corrupted) blocks, he locally corrects $y$ to find the values of $y_{j_1}, \dots, y_{j_v}$. Since the total number of errors is at most $0.1nN$, the local correcting must succeed for most copies of $y$, so taking the majority for each index, he can (with high probability) get the true values $y_{j_1}, \dots, y_{j_v}$. For each of the second $5n$ blocks, he checks how many of the bits match $y_{j_1}, \dots, y_{j_v}$. Since $j_1, \dots, j_v$ were chosen uniformly randomly and independently from the corruption, the proportion of errors in the checked bits must (by a Chernoff bound) be within $0.01$ of the true proportion of errors in this copy of $y$. Whenever the error estimated in this way is less than $0.15$, he outputs the guess for $x_i$ computed in this block. Since the total number of errors across the $5n$ copies of $y$ was at most $0.1nN$, the number of copies with at most $0.15$ fraction of errors is more than $n$, so by the end of the process we have outputted all the bits. 

% Now we describe how to decode. We pick $v = (\log n)^2$ uniformly random indices $j_1, \dots, j_v \in [N]$. Then, for the first $5n$ (corrupted) copies of $y$, we locally correct $y$ to find the values of $y_{j_1}, \dots, y_{j_v}$. Since the total number of errors is at most $0.1nN$, the local correcting must succeed for most copies of $y$, so taking the majority for each index, we can (with high probability) get the true values $y_{j_1}, \dots, y_{j_v}$. 

% Then, we will output the bits of $x$ over the next $5n$ copies of $y$. While receiving a corrupted copy of $y$, we use local decoding to figure out the next bit of $x$. In parallel, we check how many of the bits in the corrupted copy match our values of $y_{j_1}, \dots, y_{j_v}$. Since $j_1, \dots, j_v$ were chosen uniformly randomly and independently from the corruption, the proportion of errors in the checked bits must (by a Chernoff bound) be within $0.01$ of the true proportion of errors in this copy of $y$. Thus, if the proportion of errors is less than $0.19$, we can output the next bit of $x$ which we obtained, and it will be correct with high probability. On the other hand, we must (with high probability) have output a bit each time the true number of errors was at most $0.18N$.  However, since the total number of errors across the $5n$ copies was at most $0.1nN$, we can compute that the number of copies with at most $0.18N$ errors is more than $n$, so by the end of the process we have output all the bits.

The total size of the encoding is $10nN$. It turns out that we can find a Reed-Muller code with the desired properties for $N = n^{1 + O(1/t)}$, so this size is then $n^{2+O(1/t)}$, as desired.

\medskip\noindent\textbf{Bringing the error resilience up to $1/4-\e$.}
Next we outline the changes we need to make in order to bring the error tolerance up to the optimum of $1/4-\e$. 

The first change is that we use \textit{smooth decoding} to figure out the true values of $y_{j_1}, \dots, y_{j_v}$. In smooth decoding, the decoding algorithm outputs confidences of the target bit being either 0 or 1 (and the confidences add to 1). This notion was introduced as ``decoding with confidences'' in \cite{GuptaZ23}. The required lower bound in the confidence in the correct bit varies linearly by the fraction of errors, and is (roughly) 1 with no errors and 0 with a $1/2$ fraction of errors. Then, instead of tracking the majority guess for each $y_{j_t}$, Bob tracks the sum of the confidences and takes their majority. In this way, when the fraction of errors is less than $1/4-\e$, Bob will still eventually get the correct value for each $y_{j_t}$, though this point may occur very late in the stream (possibly even after up to a $1-O(\e)$ fraction of the copies of $y$ have been read).

It turns out that after this step, the proportion of errors remaining in the rest of the stream is at most $1/2-\Omega(\e)$. Now Bob will need to decode the bits of $x$ with these corrupted copies of $y$. He can use the $y_{j_1}, \dots, y_{j_v}$ as before to determine which copies do in fact have at most a $1/2-\Omega(\e)$ fraction of errors. However, the error threshold for unique decoding of $x$ is always at most $1/4$ for any (polynomial-rate) error-correcting code, so a priori it is not directly possible to get the bits of $x$ from these corrupted copies. However, the error threshold for \textit{list-decoding} can be up to $1/2$, and it turns out that we can also do a similar operation called \textit{unique local decoding with advice}. This means that we can locally decode any bit of $x$ given the true values of $y_{j'_1}, \dots, y_{j'_u}$ for $j'_1, \dots, j'_u$ drawn from some fixed distribution, not depending on which bit of $x$ we want. (After knowing $y_{j'_1}, \dots, y_{j'_u}$, we will still need to read polylogarithmically many more bits of $y$.) Moreover, this local decoding with advice has an error threshold of $1/2-O(\e)$. Thus, we can sample $j'_1, \dots, j'_u$ at the start, and figure out $y_{j'_1}, \dots, y_{j'_u}$ concurrently with figuring out $y_{j_1}, \dots, y_{j_v}$.

These changes combined yield a stream-decodable code that works up to a $1/4-\e$ fraction of errors, which is the optimal threshold for unique decoding binary codes.

\subsection{Quadratic lower bound for stream-decodable codes}
\label{subsec:overview-lb}
Next, we sketch our proof of \cref{thm:n2lower}, which states (roughly) that a stream-decodable coding scheme with decoding space $s$ resilient against any constant error proportion must use $\Omega(n^2/s)$ encoded bits. Suppose otherwise, and that a coding scheme with code length $o(n^2/s)$ exists. We will show how an adversary may corrupt an $o(1)$ proportion of the encoding to cause the decoder to fail.

Let the input $x \in \{0, 1\}^n$ be arbitrary, and split the encoding $\enc(x)$ into $k = o(n/s)$ ``blocks,'' each with $\ell = o(n)$ bits. (In reality, we will have $k=\Theta(n/s)$ and $\ell=\Theta(n)$ with very small constant factors, but for intuition it is helpful to think about them this way.) Then, we will consider what the algorithm $\dec$ might output during each block. Intuitively, this block should have at most $\ell$ bits of information, and the memory contents of the algorithm beforehand consist of $s$ bits. Though these $s$ bits may depend arbitrarily on $x$ in a way that the adversary cannot control, the $\ell$ bits are fixed, and should only allow the decoder to figure out approximately $\ell$ bits of $x$. Formalizing this intuition, we prove, roughly speaking, that for most $x$, for each block $i$, there is a set $S_i$ of size $O(\ell)$ such that the algorithm $\dec$ almost always outputs only indices (and their corresponding bits of $x$) from $S_i$ during that block, except possibly up to $O(s)$ other indices. The precise formulation of this statement is \cref{lemma:lb-1}.

To prove this lemma, we simulate the algorithm many times on block $i$, with $r=\ell/s$ different choices of memory contents $a_i^{(1)}, \dots, a_i^{(r)}$ given to it at the start of the block. Let $T_j$ be the set of indices that are output by the decoding algorithm on block $i$ when given $a_i^{(j)}$. We imagine for now that $T_j$ is deterministic; the formal proof in \cref{sec:lower} will show how to get away with this assumption. Roughly speaking, all of the index-bit pairs output over all these simulations (that is, the union of the $T_j$) must accurately match $x$, as long as the $a_i^{(j)}$ were all \textit{good} (meaning that they had high probability of success). However, the result of these simulations depends only on the choices of $a_i^{(j)}$, as well as the contents of the block, which consist of $O(\ell)$ bits total. Thus, we have a collection of index-bit pairs of $x$ with entropy at most $O(\ell)$, so their size is at most $O(\ell)$.
Now, note that we can vary the $a_i^{(j)}$ to take any collection of good memory values. We showed that under any such assignment, the union of the resulting $T_j$ has size at most $O(\ell)$. In fact, a simple combinatorial argument shows that this means that there must be some fixed set of size $O(\ell)$ which contains all but $O(\ell/r) = O(s)$ elements of each possible value of $T_j$, thus proving the lemma.
%\vgnote{Can come back to improve this if we have time, or for later version.}

Having proven this lemma, we now have a set $S_i$ of indices of size $O(\ell)$ that the algorithm can output during each block, in addition to at most $O(s)$ other indices. We then pick a uniformly random index $j \in [n/2]$ and assign to each block $i$ the interval $[j + c(i-1)s, j + cis)$, where $c$ is a sufficiently large constant. Then, for each block such that $S_i$ contains at least half the indices in that interval, the adversary sets that block to 0s thereby effectively erasing it. Since $|S_i| = O(\ell)$, each block is deleted with probability only $o(1)$ (over the randomness of $j$), so the number of deleted blocks is $o(1)$ (with probability $1-o(1)$). Then, the algorithm can never output all the elements of a block's interval during that block, since that would require outputting $cs/2$ bits not in $S_i$, which is not possible by the lemma. Thus, $\dec$ must always stay ``behind'' the intervals, and thus cannot reach the end of the input $x$.

\subsection{Stream decodable codes for linear functions of near-linear length}
\label{subsec:overview-linear-func}
For this algorithm, we'll mostly focus on constructing a length $O(n)$ code that uses $\approx \sqrt{n}$ space for stream decoding linear functions. Then, we'll give a high level outline of how to recurse our construction and reduce the space complexity to $\polylog(n)$ space at a slight cost to rate. Let Bob's secret linear function be $\ell$ (so that he wants to compute $\ell\cdot x := \ell_1x_1,\ldots,\ell_nx_n$).

\medskip\noindent\textbf{An $O(n)$ length code that uses $\approx \sqrt{n}$ space.}
    The simplest approach Bob might take to compute a linear function is to deduce $x_1,\ldots,x_n$ in order and compute $\ell_ix_i$ and add it to a running total. Unfortunately, this can't work because it runs into the lower bound of Theorem~\ref{thm:n2lower}. 
    
    Instead, Bob will try to compute the linear function in batches, or ``blocks,'' of size $\sqrt{n}$. Let us split the input into the $\sqrt{n}$ blocks $x^{1}=x_1,\ldots,x_{\sqrt{n}}~,~ x^{2}=x_{\sqrt{n}+1},\ldots,x_{2\sqrt{n}}$ and so on. Bob's goal will be to compute $\ell^{i}\cdot x^{i}$ for each $i\in [\sqrt{n}~]$.

    As a first attempt, let Alice's encoding be $\ECC(x^{1})\ECC(x^{2})\ldots$. In each block, Bob can decode $x^{i}$ and compute $\ell^{i}\cdot x^{i}$ which he adds to a running total. Unfortunately, the adversary can concentrate errors on a single $\ECC(x^{i})$, and Bob will decode incorrectly, messing up his final output.

    To address this, Alice will put an outer error-correcting code $\LDC$ on her message $\ECC(x^{1})\ECC(x^{2})\ldots$, so that she sends the message $y=\LDC \otimes \ECC(x)$. Both $\LDC$ and $\ECC$ are linear codes, so the encodings commute. Visually, picture her message bits and encoded bits as follows:
    
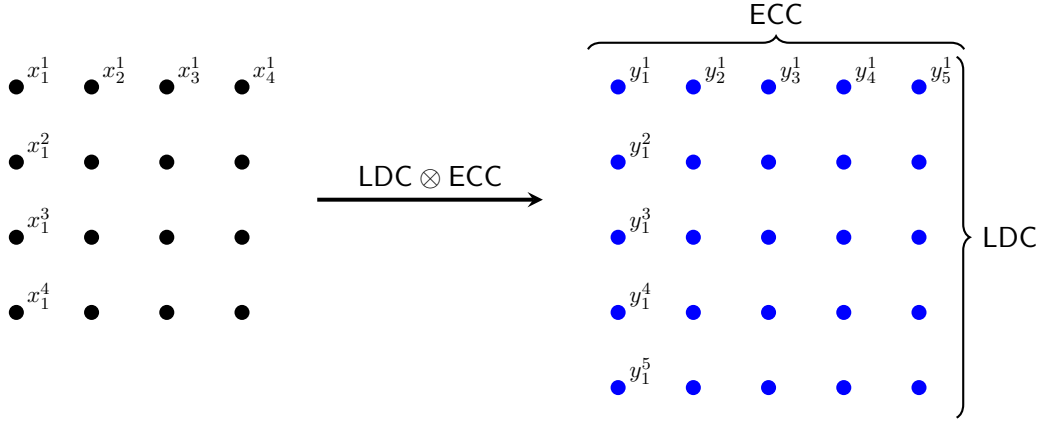
\begin{center}
\begin{tikzpicture}
  % Define the number of rows and columns
  \def \n {4}
  \def \radius {0.1cm}
  % Draw the left grid
  \foreach \i in {1,...,\n} {
    \foreach \j in {1,...,\n} {
      \fill (\i,\j) circle (\radius);
    }
  }
  % Labels for top row of left grid
  \foreach \x [count=\xi] in {1,2,...,\n} {
    \node [scale=0.8, anchor=north] at (\xi+0.3,\n+0.5) {$x^1_{\x}$};
  }
  % Labels for left column of left grid
  \foreach \x [count=\xi] in {2,3,...,\n} {
    \node [scale=0.8, anchor=north] at (1.3,\n-\xi+0.5) {$x^{\x}_{1}$};
  }
  
  % Arrow
  \draw[->, thick, line width=1.5pt, >=stealth] (\n+1,2.5) -- (\n+4,2.5) node[midway, above] {$\LDC \otimes \ECC$};

  % Draw the right grid
  \begin{scope}[shift={(8,0)}]
  \def \m {5}
    \foreach \i in {1,...,\m} {
      \foreach \j in {1,...,\m} {
        \fill[blue] (\i,\j-1) circle (\radius);
      }
    }

    % Curly brace over the top row
    \draw[decorate,decoration={brace,amplitude=5pt},xshift=-0.4cm,thick] (1,\m-0.5) -- (\m+1,\m-0.5) node[midway, above=0.2cm] {$\ECC$};
    
    % Curly brace next to the right column
    \draw[decorate,decoration={brace,amplitude=5pt},yshift=0.4cm,thick]  (\m+0.5,\m-1) -- (\m+0.5,-0.8) node[midway, right=0.2cm] {$\LDC$};

    % Labels for top row of left grid
  \foreach \x [count=\xi] in {1,2,...,\m} {
    \node [scale=0.8, anchor=north] at (\xi+0.3,\m-0.5) {$y^1_{\x}$};
  }
  % Labels for left column of left grid
  \foreach \x [count=\xi] in {2,3,...,\m} {
    \node [scale=0.8, anchor=north] at (1.3,\m-\xi-0.5) {$y^{\x}_{1}$};
  }

  \end{scope}
\end{tikzpicture}
\end{center}

The bits are sent left to right and then top to bottom. If Bob wants to compute $\ell^i\cdot x^i$ from the unencoded word $x$, he would only need to access bits of the $i$'th row. In the encoded word $y$, he can compute $\ell^i\cdot x^i$ by the following steps:
\begin{enumerate}[label=(\arabic*)]
    \item He finds a list of indices $q^i_1\ldots q^i_k$ (because the outer code is locally decodable, $k$ is small, which we will analyze later) that would recover the $i$'th index of a message encoded by $\LDC$. 
    \item He queries the $q^i_1\ldots q^i_k$'th rows of the encoded message $y$. 
    \item \textbf{Column decoding:} He uses the local decoding algorithm on each index of the rows to recover each index (with corruption) of $\ECC$.
    \item \textbf{Row decoding:} He decodes each $\ECC$ to recover row $i$ of $x$. 
    \item \textbf{Linear function:} He computes the linear function $\ell^i\cdot x^i$. 
\end{enumerate}

Because both dimensions are encoded by an error-correcting code, this is resilient to error. However, this is not actually space efficient. For any fixed index $i$, even if $k$ is small (say $\approx \polylog n$), Bob could potentially need to store all $k\cdot |\ECC| \approx \sqrt{n}$ bits all the way until the end when he gets the final row necessary to compute $\ell^ix^i$ in this manner. Since there are $\sqrt{n}$ rows, this requires $\approx \sqrt{n}\cdot \sqrt{n}$ memory. 

However, because $\ECC$ and $\LDC$ are linear, Steps (3), (4) and (5) commute (we remark that this is true as long as encoding is linear, even though decoding in the face of errors is not). In particular, he can perform them in the order (4), (5), (3). Specifically, Bob can decode each row immediately upon receiving it, then compute the linear function $\ell^j$ on each row, and only at the end use the $\LDC$ decoding algorithm on the singular bit stored for each row to compute a guess for $\ell^i \cdot x^i$.

The reason this saves space is that between rows, Bob only needs to store $k$ bits of memory per index $i$ rather than $k\sqrt{n}$ bits of memory. Since there are $\sqrt{n}$ indices, he'll need $\approx k\sqrt{n} \approx \sqrt{n}$ space.

\medskip\noindent\textbf{Recursing the construction.}
    In the current construction, Bob splits his string into $\sqrt{n}$ blocks of size $\sqrt{n}$ and computes the linear function he needs of each block separately. Instead, he could split into $n^{1/3}$ blocks of size $n^{2/3}$, and then he would need to compute a linear function on each of the blocks of size of size $n^{2/3}$. To do this, we could iterate the construction again on each of the blocks, allowing him to compute each linear function in $\approx n^{1/3}$ space instead (picture a $3$-dimensional grid instead of a $2$-dimensional grid).

    We can continue to recurse this, and the code Alice uses is $\LDC^{\otimes d}(x)$ if we recursed $d$ layers. If we recursed this, say $d=\log n/\log\log n$ times, the block length would be $\log n$. Then, at the lowest layer, Bob only needs $\log n$ space to decode the code and compute any linear functions. Of course, there is still additional space consumption from the recursive overhead and the number of small blocks. Although we won't calculate the space usage here, it is possible to get it down to $\polylog(n)$ with a slight increase in rate. 
    
    One difficulty to be aware of is that the distance of the tensor code $\LDC^{\otimes d}(x)$ is the $d$'th power of the distance of $\LDC$. This means that we will need $\LDC$ to have relative distance $1-O(1/d)$ instead of $1/2-\eps$ as was the case for most codes we have been describing so far. To deal with this, we have to embed into a larger alphabet, and also in many regimes including $d=\log n/\log\log n$, use a code where the relative distance is $1-o(1)$.

\section{Preliminaries} \label{sec:prelims}

\paragraph{Notation.}
\vspace{-2ex}
\begin{itemize}
\itemsep=0ex
    \item The function $\log$ is in base $2$ unless otherwise specified.
    \item The set $[n]$ denotes the integers $1\ldots n$.
    \item Given a tuple $T$ and element $i$, the expression $T|i$ denotes $i$ concatenated to $T$.
    \item The phrase ``with high probability in $n$'' means with probability at least $1-\frac{1}{n^{\omega(1)}}$.
    \item We use $\Delta(x,y)$ to denote the Hamming distance between two strings $x,y\in (\Sigma \cup \bot)^n$, and $\delta(x,y)$ to denote the relative distance between them (i.e. $\delta(x,y)=\frac1n\cdot \Delta(x,y)$). Any element of $\Sigma$ is considered distance $\frac12$ from $\bot$.
    \item For clarity, we will often omit floor and ceiling signs where they would technically be necessary.
\end{itemize}

\begin{lemma}[Tail bound for $k$-wise independent random variables]
\label{lem:k-wise-tail}
Let $k > 4$ be an even integer. Suppose 
$X_1,X_2,\dots,X_n$ are $k$-wise independent
random variables taking values in $[0, 1]$. Let $Z = \sum_i X_i$ and $\mu =\E[Z]$. Then 
\[ \mathrm{Pr} [ |Z-\mu| \ge A] \le 8 \cdot \Bigl( \frac{k \mu + k^2}{A^2}\Bigr)^{k/2}  \ . \]

\end{lemma}
\subsection{Error-correcting codes} \label{sec:ecc}

We begin with some results about error-correcting codes. We first state a theorem detailing the existence of distance $1-1/|\bbK|-\eps$ codes that are efficiently encodable and efficiently decodable. It is standard, and based on GMD decoding of concatenated codes with an outer code approaching the Singleton bound (like Reed-Solomon or algebraic-geometric codes), and a small inner code of relative distance close to $(1-1/|\bbK|)$ (see for instance \cite[Chap.14]{GuruswamiRS19}).

\begin{theorem}
\label{thm:binary-ecc}
    For every $\eps > 0$ and every finite field $\bbK$,  there exists an explicit systematic linear error-correcting code $\ECC_\eps = \{ \ECC_{\eps, n} : \bbK^n \rightarrow \bbK^m \}_{n \in \bbN}$ with relative distance at least $1-1/|\bbK|-\eps$ and $m \le n/\eps^{O(1)}$, and a $O_\eps(n^2)$-time and space decoding algorithm $\DEC_{\eps} : \bbK^m \rightarrow \bbK^n$, such that for any $x \in\bbK^n$ and $w \in \bbK^m$ satisfying $\delta(\ECC_\eps(x), w) < (1 - 1/|\bbK|)(1-\eps)/2$, it holds that $x = \DEC_\eps(w)$.
\end{theorem}

Our constructions will use Reed-Muller codes (based on evaluations of multivariate polynomials) concatenated with the codes from Theorem~\ref{thm:binary-ecc}. 
In order to locally decode these codes, we will correct them along lines for which we would need to run list decoding algorithms for concatenated codes with outer Reed-Solomon codes. The following list decoding result for such codes is standard, and based on list-decoding the inner codes by brute force and list-recovering the outer Reed-Solomon codes; see for example~\cite{GuruswamiS00}. (Better parameters are possible, but having $\poly(\eps)$ rate and $\poly(1/\eps)$ output list size suffices for our purposes.)

\begin{theorem}
\label{thm:concat-list-decode}
 Let $\eps> 0$. Let $\C$ be a concatenated code with outer Reed-Solomon code over $\F_q$ of rate $(\frac\eps4)^4$ and an inner code of relative distance at least $\tfrac{1}{2} -\tfrac{\eps^2}{16}$. Then $C$ can be list-decoded in $\poly(q)$ time from a fraction $(1-\eps)/2$ of errors with an output list size of $64/\eps^3$.
\end{theorem}

\iffalse
%%% If needed we can put this proof in.
   We describe a pretty standard approach to list decode concatenated codes (quantitatively better algorithms are known~\cite{GuruswamiS00} but are not necessary for us).  We list decode the $\lambda$'th block $w|_{P,\lambda}$ of $w|_P$, corresponding to $\lambda \in \F \setminus \{0\}$, up to a radius of $(1-\eps)/2$ outputting a set $S_\lambda \subset \F$ consisting of all field elements $\alpha$ such that $\delta(C_{inner}(\alpha), w|_{P,\lambda}) \le (1-\eps)/2$. This can be done by brute force in $\poly(|\F|) \le \poly(Q)$ time.  By Theorem~\ref{thm:johnson}, $|S_\lambda| \le 2/\eps^2$. 
   \fi

\section{Locally decodable/correctable codes}
%\vgnote{It seems better to have this be a separate section and not clubbed into preliminaries. We can also move the proofs of Theorems 3.4 and 3.9 into an appendix.}
In this section, we introduce the locally decodable/correctable codes that will form the backbone of the constructions in our paper. There are two theorems we require, one for Section~\ref{sec:n2-const} and one for Section~\ref{sec:linear}. Each of our codes will require a feature besides just correctness of local decoding/correcting when the distance to a codeword is small.

Our first code for binary alphabets has two additional requirements. It requires that the decoder output a probability of each output $0$ and $1$ rather than only one of them (smoothness). This requirement is similar to list decoding with confidences from~\cite{GuptaZ23}, and we adapt their proof below. Secondly, it has a local decoding with advice guarantee. To establish this notion, we use ideas similar to \cite{SudanTV99} and the locally list-decodable codes of~\cite{GoldwasserGHKR07}.

Our second code requires only a smoothness guarantee for local decoding. However, it is in the regime with large alphabet and non-constant $\eps$, and the code is required to be linear.

\begin{theorem} (Binary locally decodable code) \label{thm:ldc-bin}
    Fix an arbitary $\eps > 0$. Let $Q = Q(n) \in [(\log n)^{100},n]$. There is a code $\LDC : \{0,1\}^n \rightarrow \{0,1\}^{N}$ 
    that satisfies the following properties:
    \begin{itemize}
    \item {\bf Length:} 
        The length $N =  N(n) \le n\cdot (\log_Q n)^{100\log_Q n}$.
    \item {\bf Distance:} 
        For any $x \not= y \in \{0,1\}^n$, it holds that $\delta(\LDC(x), \LDC(y)) \ge \frac12 - \eps$.
        
    \item {\bf Smooth local decoding/correcting:} 
        There exists a randomized algorithm $\cA$ that on input $i\in [n]$ (resp. input $i\in [N]$) non-adaptively queries $Q$ bits of the encoding and runs in $O(Q^3)$ time and space and achieves the following guarantee. For any word $w \in \{ 0, 1 \}^N$,  with high probability in $n$, the algorithm outputs probabilities $p(b)$ for $b\in \{0,1\}$ that satisfy $p(0)+p(1)=1$ and $p(x_i)>1-2\delta(w, \LDC(x))-\eps$ (resp. $p(\LDC(x)_{i})>1-2\delta(w, \LDC(x))-\eps$).
        %Furthermore, $p(0)$ and $p(1)$ can be represented as rational numbers with denominator $4Q$. 

    \item {\bf Local decoding with advice:}
        There exists a randomized algorithm that on input $i\in [n]$ queries $Q$ bits of the encoding (non-adaptively)  and runs in $\poly(Q)$ time and space, and a distribution $\cD$ on $[N]^u$ (independent of $i$) for some $u=O(Q)$ (that is, subsets of indices of size $u$), which does the following. For any word $w \in \{ 0, 1 \}^N$ satisfying $\delta(w, \LDC(x))<\frac12-\eps$, it outputs $x_i$ with high probability in $n$ when additionally given $\LDC(x)_{d_1}\ldots \LDC(x)_{d_u}$ for $d_1\ldots d_s \sim \cD$. 
    \end{itemize}
\end{theorem}

\begin{proof}
    The locally decodable code $\LDC$ that we will be using is the concatenation of two codes. The outer code $\C_{outer}$ is a Reed-Muller code with the following parameters. Recall that the length of the message in binary is $n$, and the number of queries permitted is $Q$. We let $d=\eps^6\sqrt{Q}/4$ be the degree of the polynomials, $q=\sqrt{Q}$ be the field size and $m$ be the smallest integer such that $n \le \binom{d+m}{m}$. We assume all these variables are integers, and that $q$ is a power of $2$. The inner code $\C_{inner}$ is a linear alphabet $\bbF_2$ code of relative distance $\frac12 - \frac{\eps^6}{4}$ and rate $\Omega_\eps(1)$, as guaranteed by Theorem~\ref{thm:binary-ecc}.
    Throughout the proof, we will assume $n$ is sufficiently large compared to $\eps^{-1}$ which is treated as a constant. 

    We will split the proof into four claims, detailing the rate of the code, distance of the code, the smooth local decoding property, and the local decoding with advice property.

\begin{claim}
    The code $\LDC$ has length at most $n\cdot \log_Q(n)^{O(\log_Q(n))}$.
\end{claim}

\begin{proof}
    We will show that $m<10\log_Q n$. Since $\binom{d+m}{m}$ is increasing with $m$, it suffices to show that $\binom{d+10\log_Q n}{10\log_Q n}> n$. Noting that $Q^{0.1}>100\eps^{-6}\log n$ since $Q>(\log n)^{100}$ and that $Q$ is sufficiently large compared to the constant $\eps^{-1}$,
    \begin{align*}
        &\binom{d+10\log_Q n}{10\log_Q n} > \left(\frac{d}{10\log_Q n} \right)^{10\log_Q n} > \left(\frac{\eps^6\sqrt{Q}}{100\log_Q n} \right)^{10\log_Q n} \\
        > &\left(\frac{Q^{0.4}\log n}{\log_Q n} \right)^{10\log_Q n} > Q^{4\log_Q n} > n^{4} \ .
    \end{align*}
    
    We use this to calculate an upper bound on the length of the outer code. The expression $\frac{N(n)}{n}$ is $\frac{q^m}{\binom{d+m}{m}}$, which is upper bounded as follows:
    \begin{align*}
        &\frac{q^m}{\binom{d+m}{m}} \leq \frac{q^m\cdot m^m}{d^m} = \left(\frac{qm}{d}\right)^m = (4m\eps^{-6})^m < m^{2m}\\ 
        <~ &(10\log_Q n)^{20\log_Q n} < (\log_Q n)^{50\log_Q n} \ .
    \end{align*}
%\vgnote{Something off in the last expression - why do we even need it, as we already proved $N \le n \cdot (\log_Q n)^{50\log_Q n}$ which is better than what's promised in the theorem statement} \mgnote{Oops, resolved.}
    Here, we use that for sufficiently large $m$, it holds that $m>4\eps^{-6}$.
    The inner code $\C_{inner}$ is constant rate $\Omega_\eps(1)$, and thus the overall code length is upper bounded by $n \cdot \log_Q(n)^{100\log_Q n}$.
\end{proof}

\begin{claim}
\label{clm:distance1}
    The code $\LDC$ has relative distance at least $\frac12-\eps^6$.
\end{claim}

\begin{proof}
    By the Schwartz-Zippel lemma~\cite{Schwartz80,Zippel79}, the relative distance of the outer code is at least $(q - d + 1)/q  \ge \left( 1 - \frac{\eps^6}2 \right)$. So the relative distance of the concatenated code is $\ge \left( 1 - \frac{\eps^6}{2} \right) \cdot \left( \frac12 - \frac{\eps^6}4 \right) \ge \frac12 - \eps^6$.
\end{proof}

\begin{claim}
\label{clm:smooth-unique-decoding1}
    The code $\LDC$ satisfies the  smooth local decoding property.
\end{claim}

\begin{proof}
  %  The proof is similar to the large alphabet case, but a bit simpler. We start with the same definitions as before.
   % 
    Denote the length of the outer code as $N_{outer}$ and the length of the inner code as $N_{inner}$. We will index the symbols of the concatenated code by pairs $(v, j)$ where $v \in \bbF^m$ and $j \in [N_{inner}]$. We remark that the inner code $\C_{inner}$ and outer code $\C_{outer}$ are systematic, so local correcting is sufficient.

    The decoding algorithm is as follows. 
    Suppose we want the $i$'th index of $x$ which is located at $v_0 \in \bbF^m$ in the outer code. Suppose that it is the $a$'th index of this symbol. Then, we do the following $t = \sqrt[3]{Q}$ times: Pick a random degree-$2$ cruve $P(\lambda) = v_0 + v_1 \lambda + v_2 \lambda^2$ through $v_0$ by sampling $v_1, v_2 \gets \bbF^m$. Then, we query all symbols of $w$ located at $(P(\lambda), j)$ for $\lambda \in \bbF^*$, $j \in [N_{inner}]$ (we denote these collection of values by $w|_P$). Then the number of queries is $t \cdot (q-1) \cdot N_{inner} = t\cdot (q-1)\cdot O_\eps(\log q) < Q$. By Theorem~\ref{thm:concat-list-decode} (or by standard GMD decoding using that both the inner and outer codes are efficiently decodable), one can in $O(q^2) < O(Q^2)$ time and space find the unique degree $2d$ polynomial $h \in \bbF[\lambda]$, if it exists, such that the relative distance between $w|_P$ and the $\C_{inner}$ encodings of $h(\lambda)$ is at most $\frac12 \cdot \frac{(q - 1-2d)}{q-1} \cdot \left( \frac12 - \frac\eps2 \right) \ge 1/4-\eps$.
    
    For each of the $t$ polynomials $P$, let $b:=C_{inner}(h(0))_a$, that is, the $r$'th symbol of the $C_{inner}$ encoding of $h(0)$, representing a candidate for the $i$'th index of the codeword. We let $\delta^{b}_P$ be the relative distance between the $\C_{inner}$ encoding of $h$ and $w|_P$, i.e. the queried symbols from the $\C_{inner}$ encodings of $b$. Let $\delta^{1-b}_P=\frac12-\delta^{b}_P$, or if $h$ doesn't exist, let $b=0$ and $\delta^0_P=\delta^1_P=\frac14$. Then, the decoding algorithm outputs 
    \[
        p(0)=\frac{2}{t}\cdot \sum_P \delta^1_P \quad \text{and} \quad p(1)=\frac{2}{t}\cdot \sum_P \delta^0_P.
    \]
    These values clearly satisfy the requirement $p(0)+p(1)=1$. 

    We now show that the smooth local decoding property holds. We have that
    \begin{align*}
        \Pr& \left[ p(x_i) \le 1 - 2\delta(w, \LDC(x)) -\eps  \right] \\
        &= \Pr \left[ 1-p(x_i) \geq 2 \delta(w, \LDC(x))+\eps\right] \\
        &\le \Pr \left[ \frac{2}{t}\cdot \sum_P \delta^{x_i}_P \ge 2\delta(w, \LDC(x)) + \eps \right] \\
        &\le \Pr \left[ \frac{1}{t}\cdot \sum_P \delta^{x_i}_P \ge \delta(w, \LDC(x)) + \frac\eps2 \right].
    \end{align*}
    
    Also, we have that $\delta^{x_i}_P \le \delta(\LDC(x)|_P, w|_P)$, so that $\bbE \left[ \frac{1}{t}\cdot\sum_P \delta^{x_i}_P \right] \le \delta(w, \LDC(x))$. Then, by Hoeffding's inequality, 
        \[
            \Pr \left[ \frac{1}{t}\cdot \sum_P \delta^{x_i}_P \ge \delta(w, \LDC(x)) + \frac\eps2 \right]
            < e^{-\eps^2 t/2} \ .
        \]

    As long as $Q$ is sufficiently large relative to $\eps^{-1}$, it holds that $\eps^2t/2=\eps^2\sqrt[3]{Q}>(\log n)^2$, so the smooth local decoding succeeds with high probability. The overall space taken for all the iterations is at most $O(Q^3)$.
\end{proof}

\begin{claim}
\label{clm:advice-ld}
    The code $\LDC$ satisfies the local list decoding with advice property.
\end{claim}

\begin{proof}
Recall that we are given (oracle access to) $w \in \{0,1\}^N$ where $\delta(w,\LDC(x)) \le \frac12-\eps$, and we would like to locally decode $x_i$ with high confidence making $Q$ queries to $w$, and using as advice $s := N_{inner}/\eps$ random bits from the codeword $\LDC(x)$. These bits will be the bits at codeword positions randomly sampled according to a distribution $\cD$ on $[n]^s$ defined  as follows. We choose $k := 1/\eps$ uniformly random indices of the outer Reed-Muller code, and include all corresponding $N_{inner}$ indices of the inner code
    for a total of $s$ indices of the overall code.  
    
% \vgnote{Might be good to give a pseudocode for the algorithm. Right now the algorithm is interspersed with the ideas and analysis. Probably not a priority now given the time...}

% \vgnote{IMPORTANT: Should credit~\cite{STV01} for the general idea and compare. I forget but is our task easier we can sample an advice separately for each $x_i$?}

Our approach for decoding with advice closely follows the approach used in \cite{STV01} for locally list-decoding Reed-Muller codes.

    Let $v_1, v_2, \ldots v_k$, with each $v_i \in \F^m$, be the indices of the outer Reed-Muller code drawn from $\cD$. 
    Suppose we want to decode the $i$'th index of $x$ which is located at $v_0 \in \bbF^m$ in the outer code. Our code is systematic, so it suffices to be able to decode indices of the encoded word.
    The decoding proceeds as follows. 
    
    Choose $j_1, j_2, \ldots j_{k}$ to be random distinct points in $\F^*$ (the nonzero elements of $\F$). Let $P(\lambda)$ be the degree-$k$ polynomial curve $\F \to \F^m$ satisfying $P(0)=v_0, P(j_1)=v_1, \ldots, P(j_{k})= v_{k}$. Our goal will be to recover the degree-$dk$ univariate polynomial $h:=\C_{outer} \circ P$. (We are thinking of $\C_{outer}$ as an $m$-variate degree-$d$ polynomial here, and when composed with the curve $P$, it gives a univariate polynomial.) We query all bits of the string $w \in \{0,1\}^N$ to be list decoded that are located at $(P(\lambda), j)$ for $\lambda \in \F^*$, $j \in [N_{inner}]$ (we denote these collection of values by $w|_P$). Then the number of queries is $(q-1) N_{inner} \le q \log q \cdot \eps^{-O(1)} < Q$, since $n$ and therefore $Q$ is large enough compared to $1/\eps$.

Note that if there are no errors, $w|_P$ should equal the encoding of $h$ by code $\C^P_{concat}$ which is the concatenation of the RS code (for degree $dk$ and evaluation set $\F^*$) with $\C_{inner}$. 
We will now argue that the randomness of $P$ implies that with high probability $w|_P$ is within relative distance $\frac12-\eps$ of this encoding $\C^P_{concat}(h)$. 

For $v \in \F^m$, let $\eps_v$ denote the fraction of bits at locations $(v;j)$, $j \in [N_{inner}]$, where $w$ differs from $\LDC(x)$. We have $\rho := \E_{v \in \F^m} [\eps_v] \le (\frac12-\eps)$. Since $v_1,v_2,\dots,v_k$ are chosen at random from $\F^m$, the random variables $P(\lambda)$, $\lambda \in \F^*$ are $k$-wise independent. 
Therefore, for a random $\lambda \in \F^*$, $\eps_{P(\lambda)}$ will concentrate around $\rho$. Specifically, applying Lemma~\ref{lem:k-wise-tail} to the random variables $\eps_{P(\lambda)}$, $\lambda \in \F^*$, shows that 
\begin{equation}
\label{eq:local-line-good}
    \mathrm{Pr}_P \Bigl[ \ \E_{\lambda \in \F^*} [ \eps_{P(\lambda)} ] \ge \frac{1-\eps}{2} \ \Bigr] \le  \Bigl( \frac{O(k)}{\eps q}\Bigr)^{k/2} \ . 
\end{equation}

 Using Theorem~\ref{thm:concat-list-decode}, we can now list decode $w|_P$ up to a radius $(1-\eps)/2$ to get a list of $L$ polynomials $h_1,h_2,\dots,h_L$, for $L \le \poly(1/\eps)$. By  \eqref{eq:local-line-good} and recalling that $k = 1/\eps$, the correct polynomial 
$h$ belongs to this list except with probability $(\eps q)^{-\Omega(1/\eps)}$. 
We will henceforth assume this is the case.

It remains to describe how to uniquely identify $h$ amongst the list. The algorithm will use the advice for this purpose. Specifically, the algorithm will check if there is a unique polynomial $h_i$ in the list that satisfies $h_i(j_{\ell}) = \LDC(x)[v_\ell]$ for $\ell = 1,2,\dots,k$. Here, for $v \in \F^m$, $\LDC(x)[v] \in \F$ denotes the field element whose encoding equals the advice bits at locations $(v;j)$, $j \in [N_{inner}]$. If this is the case, the algorithm will compute $h(v_0)$ and output the appropiate bit $(v_0;j)$ from the inner encoding that corresponds to $x_i$.

By definition, $h$ will pass this check. For $h_i \neq h$, since they are both degree $kd$ polynomials, they agree on at most $kd$ elements of $\F$. 
%\vgnote{Double check validity of this argument}
Since the points $v_1\ldots v_{k}$ are chosen uniformly at random, the polynomial $P$ is a random degree ${k}$ polynomial satisfying $P(0)=v_0$. We can therefore reframe the algorithm as follows. First, pick a random polynomial $P$ satisfying $P(0)=v_0$. Then, pick ${k}$ random points $j_1\ldots j_{k}$ to be random distinct points in $\F^*$, and set $v_1=P(j_1),\ldots, v_k=P(j_{k})$ (this results in the same algorithm as above where first picked $v_1\ldots v_{k}$ and then generated $P$). Then, since $j_1\ldots j_{k}$ are uniformly random, the probability that all those indices of $h$ match $h_i$ is at most $(\tfrac{kd}{|\F|})^k \le \eps^{3/\eps}$. 
%\vgmargincomment{Does it need adjustment since $j_i$'s are sampled without replacement.} 
Union bounding over all $L \le \poly(1/\eps)$ of the $h_i$'s, the probability that no $h_i\neq h$ passes the above tests at locations $j_1 \ldots j_{k}$ is at least $1-\eps^{\Omega(1/\eps)}$. %\vgnote{It seems just advice at just $O(1)$ points suffices to disambiguate with success prob. $1-\eps$. Should we bother to change?}

  %  Now, we show that the advice bits permit us to identify which polynomial of $h_1\ldots h_e$ is equal to $h$ with probability at least $1-\frac1{2^{k}}$. In particular, we will show that with this probability, a given $h_i \neq h$ does not align with $h$ on all the indices $j_1\ldots j_{k}$. First, note that if a polynomial $h_i\neq h$, then it overlaps on at most $kd$ locations in $\bbF$. 

    In all, the probability that both the correct $h$ belongs in the list $h_1\ldots h_L$, and no other $h_i$ is mistakenly chosen to be correct is at least $1-\eps^{\Omega(1/\eps)}>0.9$.
    %\vgmargincomment{Lot of slack here!}
Then, we repeat this algorithm $t$ times. Each instance of the algorithm will identify the correct value of $x_i$ with probability at least $0.9$, so after $t := \sqrt{Q} > (\log n)^{\Omega(1)}$ iterations, the majority output for $h$ will be correct with high probability in $n$.
\end{proof}  
   
%    By Theorem~\ref{thm:concat-RS-decoding} (combined with Theorems~\ref{thm:binary-ecc} and~\ref{thm:berlekamp-welch}), one can in $\poly(q) = \poly(Q)$ time and $O(Q^2)$ space find an $\frac1\eps$-sized list  of all the degree $d/\eps$ polynomials $h_1\ldots h_e \in \bbF[\lambda]$ such that the distance between $w|_P$ and the $\C_{inner}$ encodings of $h(\lambda)$ is at most $(1-\eps) \cdot (q - d/\eps) \cdot \left(\frac12-\frac12\eps\right) \ge \left( \frac12 - \eps \right) $.

The proof of Theorem~\ref{thm:ldc-bin} is complete.
\end{proof}

We next turn to the proof of the large alphabet version of Theorem~\ref{thm:ldc-bin}. Here, we will only need the smooth local decoding guarantee.

%\vgnote{Why such a high power like 100 below?}
\begin{theorem} (Large alphabet locally decodable code) \label{thm:ldc-large}
    Let $\eps > 0$ and let $\bbK$ be a field of the form $\bbF_{2^k}$ where $2^k>\eps^{-10}$, and let $Q = Q(n) \in [ (\eps^{-1}\log n)^{100}, n]$. There is a linear code $\LDC : \bbK^n \rightarrow \bbK^{N}$
    that satisfies the following properties:
    \begin{itemize}
    \item {\bf Length:} The length $N =  N(n)$ satsifies $N \le n\cdot (\eps^{-1}\log_Q(n))^{100\log_Q(n)}$.
    %and $N=\poly(\eps^{-1}n)$.
    \item {\bf Distance:} 
        For any $x \not= y \in \bbK^n$, it holds that $\delta(\LDC(x), \LDC(y)) \ge 1- \eps^6$.

    \item {\bf Large alphabet smooth local decoding/correcting:} 
        There exists a randomized algorithm $\cA$ that on input $i\in [n]$ (resp. input $i\in [N]$) queries $Q$ bits (non-adaptively) of the encoding and runs in $O(Q^3)$ time and space and does the following. For any word $w \in (\bbK \cup \bot)^N$, it outputs a list of probabilities $p(\sigma^*)$ for $\sigma^*\in (\bbK \cup \bot)$, satisfying that $\sum_{\sigma^*\in (\bbK\cup\bot} p(\sigma^*)=1$, at most one $\sigma\in\bbK$ has $p(\sigma)>0$, and $p(x_i)+0.5p(\bot)>1-\delta(w, \LDC(x))-\eps$ (resp. $p(\LDC(x)_{i})+0.5p(\bot)>1-\delta(w, \LDC(x))-\eps$) with high probability in $n$. Here, the Hamming distance between $\bot$ and $\sigma\in \bbK$ is $0.5$.
        Moreover, the decoding algorithm queries any specific index with probability at most $\frac{1.1Q}{N}$.
        
    \end{itemize}
\end{theorem}

\begin{proof}
    The construction is very similar to that of Theorem~\ref{thm:ldc-bin} with some adjustments for the large alphabet.
    
    The locally decodable code $\LDC$ that we will be using is the concatenation of two codes. The outer code $\C_{outer}$ is a Reed-Muller code with the following parameters. Recall that the message has $n$ symbols from $\bbK$. The number of queries permitted for the local decoding is $Q$. We let $q<\sqrt{Q}$ be the field size (moreover, let $q$ be a multiple of $2^k$), let $d=\eps^6q/4$ be the degree of the polynomials, and let $m$ be the smallest integer such that $n \le \binom{d+m}{m}$. We assume all these variables are integers, and that $q$ is a power of $2$. The inner code $\C_{inner}$ is an alphabet $\bbK$ linear code of relative distance $1- \frac{\eps^6}{4}$ and rate $\poly(\eps^{-1})$, as guaranteed by Theorem~\ref{thm:binary-ecc}. The concatenation is done via a linear map to make $\LDC$ linear.

    We will split the proof into three claims, detailing the rate of the code, distance of the code, and the smooth local decoding property.

\begin{claim}
    The code $\LDC$ has length at most $n\cdot (\eps^{-1}\log_Q(n))^{100\log_Q n}$.
\end{claim}

\begin{proof}
    The proof is essentially identical to the previous construction, but we reproduce it since in this construction $\eps$ is no longer treated as a constant. We will show that $m<10\log_{Q} n$. Since $\binom{d+m}{m}$ is increasing with $m$, it suffices to show that $\binom{d+10\log_{Q} n}{10\log_{Q} n}> n$. Noting that $Q^{0.1}>100\eps^{-7}\log n$,
    \begin{align*}
        &\binom{d+10\log_{Q} n}{10\log_{Q} n} > \left(\frac{d}{10\log_{Q} n} \right)^{10\log_{Q} n} > \left(\frac{\eps^6\sqrt{Q}}{100\log_{Q} n} \right)^{10\log_{Q} n} \\
        > &\left(\frac{Q^{0.4}\log n}{\log_{Q} n} \right)^{10\log_{Q} n} > (Q)^{4\log_{Q} n} > n^{4} \ .
    \end{align*}
    
    We use this to calculate an upper bound on the length of the outer code. The expression $\frac{N(n)}{n}$ is $\frac{q^m}{\binom{d+m}{m}}$, which is upper bounded as follows:
    \begin{align*}
        &\frac{q^m}{\binom{d+m}{m}} \leq \frac{q^m\cdot m^m}{d^m} = \left(\frac{qm}{d}\right)^m = (4m\eps^{-6})^m < (\eps^{-1}m)^{6m} \\
        <~ &(\eps^{-1}10\log_{Q} n)^{60\log_{Q} n} < (\eps^{-1}\log_{Q} n)^{100\log_{Q} n}\ .
    \end{align*}

    The inner code $\C_{inner}$ has rate lower bounded by a polynomial in $\eps$, and thus the overall code length is upper bounded by $n \cdot (\eps^{-1}\log_Q n)^{100\log_Q n}$.
    % Now, we bound this in two ways. Firstly, $(\eps^{-1}\log_{Q} n)^{100\log_{Q} n} < (\eps^{-1}\log_Q n)^{O(\log_Q n)}$ since $Q>Q$. Secondly, 
    % \begin{align*}
    %     &~(\eps^{-1}\log_{Q} n)^{100\log_{Q} n} < (1/\eps)^{100\log_{Q} n} \cdot \log_{Q} (n)^{100\log_{Q} n} \\
    %     <&~ (Q)^{100\log_{Q} n} \cdot (\log n)^{100\log_{Q} n} < n^{100}\cdot n^{100}
    % \end{align*}

    % Here, we use that $Q>\log n$. Thus the overall code length is upper bounded by $n \cdot (\eps^{-1}\log_Q n)^{O(\log_Q n)}$.
\end{proof}

\begin{claim}
\label{clm:distance2}
    The code $\LDC$ has relative distance at least $1-\eps^6$.
\end{claim}

\begin{proof}
    By the Schwartz-Zippel lemma~\cite{Schwartz80,Zippel79}, the relative distance of the code is at least $(q - d + 1)/q \cdot \left( \frac12 - \eps^6/4 \right) \ge 1 - \eps^6 $.
\end{proof}

\begin{claim}
\label{clm:smooth-unique-decoding2}
    The code $\LDC$ satisfies the smooth local decoding property.
\end{claim}

\begin{proof}
    Denote the length of the outer code as $N_{outer}$ and the length of the inner code as $N_{inner}$. We will index the symbols of the concatenated code by pairs $(v, j)$ where $v \in \bbF^m$ and $j \in [N_{inner}]$. We remark that the outer code $\C_{outer}$ is systematic, so smooth local correcting is sufficient.

    The decoding algorithm is as follows. 
    Suppose we want the $i$'th index of $x$, which is located at $v_0 \in \bbF^m$ in the outer code. It is the $a$'th index of this symbol. Then, we do the following $t = \sqrt[3]{Q}$ times: Pick a random degree-$2$ polynomial $P(\lambda) = v_0 + v_1 \lambda + v_2 \lambda^2$ by sampling $v_1, v_2 \gets \bbF^m$. Then, we query all symbols of $w$ located at $(P(\lambda), j)$ for $\lambda \in \bbF^*$, $j \in [N_{inner}]$ (we denote these collection of values by $w|_P$). Then the number of queries is $t \cdot (q-1) \cdot N_{inner} = t\cdot (q-1)\cdot \log(q) < Q$. Next, we will in $O(q^2) < O(Q^2)$ time and space find the unique degree $2d$ polynomial $h_P \in \bbF[\lambda]$, if it exists, such that the relative distance between $w|_P$ and the $\C_{inner}$ encodings of $h_P(\lambda)$ is at most $\frac12 -\frac\eps2 \leq \frac12 \cdot \frac{(q - 1 - 2d)}{q-1} \cdot \left( 1 - \frac\eps4 \right)$ (note that it's unique even with the erasure symbols in $w$). This can be done, for example, by setting every instance of $\bot$ to $0$ and performing list decoding within radius $1-\frac\eps2$, using Theorem~\ref{thm:concat-list-decode} to obtain a list of size $O(1)$ candidates, and then checking each manually.

    For each of the $t$ polynomials $P$, let $b_P:=\C_{inner}(h_P(0))_a$, that is, the $a$'th symbol of the $\C_{inner}$ encoding of $h_P(0)$, representing a candidate for the $i$'th index of the codeword. 
    %\vgnote{The notation $p_P$ is a bit odd - should we call the polynomial $R$ and change to $p_R$ or something?}
    Next, set $p_P(b)=1-2\delta(w|_{P},\C_{inner} \text{ encoding of } h_P)$,  
    %\vgnote{Should we introduce and say something like $C_{\text{concat}}(h_P)$ instead of inner encodings of (what should be the evaluations of) $h_P$? This also applies to the earlier binary code proof, and couple of uses in the next paragraph too. We can even introduce notation like $z_P = C_{\text{concat}}(h_P)$}
    $p_P(\bot)=1-p^b_P=2\delta(w|_{P},\C_{inner}$, and for all other $\sigma\in \bbK$, set $p_P(\sigma)=0$. If $b_P$ does not exist, just set $p_P(\bot)=1$ and $p_P(\sigma)=0$ for all $\sigma\in\bbK$. Then, for all $\sigma^*\in (\bbK \cup \bot)$, the decoding algorithm outputs $p(\sigma^*)=\frac1t\sum_P p_P(\sigma^*)$.
    These values satisfy the requirement $\sum_{\sigma^*}p(\sigma^*)=1$. If multiple values of $p(\sigma)$ for $\sigma\in \bbK$ are nonzero, simply decrease both by the minimum of the two and increase $p(\bot)$ by twice as much. The smooth local decoding inequality will still be true if we can show it for the original values.

    Before we prove the smooth local decoding property, we will show that $p_P(x_i)+0.5p_P(\bot) \ge 1-\delta(w|_P, \LDC(x)|_P))-\frac\eps2$. For $h_P$, it holds that $p_P(b_P) + 0.5p_P \ge 1-\delta(w|_{P},\C_{inner} \text{ encoding of } h_P) - \frac\eps2 $. Otherwise, for all other potential codewords $z$ restricted to $P$, if $b_P$ exists, it is case that 
    \[
        \delta(w|_P,z)\geq 1-\frac\eps2-\delta(w|_{P},\C_{inner} \text{ encoding of } h_P) \geq 1- \frac\eps2 - 0.5p_P(\bot)
    \]
    which implies that $0.5p_P(\bot) \ge 1-\delta(w|_P,z)-\frac\eps2$. On the other hand, if $b_P$ doesn't exist, then 
    \[
        p_P(x_i)+0.5p_P(\bot) \geq 0+0.5\cdot 1 = 0.5 \geq 1-\delta(w|_P, \LDC(x)|_P))-\frac\eps2.
    \]
    Thus, it holds that that $p_P(x_i)+0.5p_P(\bot) \ge 1-\delta(\LDC(x)|_P), w|_P)-\frac\eps2$.

%\vgnote{The structure of the argument is slightly different from Claim \ref{clm:smooth-unique-decoding1}; it would be good to uniformize.}
    We now discuss why the smooth local decoding property holds. We have that
    \begin{align*}
        \Pr& \left[ p(x_i)+0.5p(\bot) \le 1 - \delta(w, \LDC(x)) -\eps  \right] \\
        &\le \Pr \left[ \frac1t\sum_P \left(1-\delta(\LDC(x)|_P), w|_P)\right) \le 1 - \delta(w, \LDC(x)) -\frac\eps2\right] \\
        &= \Pr \left[ \frac1t\sum_P \delta(\LDC(x)|_P), w|_P) \ge \delta(w, \LDC(x)) +\frac\eps2\right] .
    \end{align*}
    
    By Hoeffding's inequality,
        \[
            \Pr \left[ \frac1t\sum_P \delta(\LDC(x)|_P), w|_P) \ge \delta(w, \LDC(x)) +\frac\eps2\right]
            < e^{-\eps^2 t/2} \ .
        \]

    As long as $Q$ is sufficiently large relative to $\eps^{-1}$, it holds that $\eps^2t/2=\eps^2\sqrt[3]{Q}>(\log n)^2$, so the smooth local decoding succeeds with high probability. The overall space taken to compute all the iterations is at most $O(Q^3)$.
\end{proof}

The proof of Theorem~\ref{thm:ldc-large} is now complete.
\end{proof}
\section{Stream decodable codes of near quadratic length} \label{sec:n2-const}

In this section we prove \cref{thm:n2const}, which is a construction of a stream-decodable code that achieves nearly quadratic length in polylogarithmic space. We restate it here:

\quadconst*

For convenience, we will scale $\e$ by a factor of 10, so that the adversary introduces at most $(1/4 - 10\e)m$ errors into the stream (rather than $(1/4-\e)m$). Also, we will present an algorithm whose space is $O(s(n))$ rather than just $s(n)$, since we can simply scale $s(n)$ to account for this. We will assume throughout this section that $n$ is sufficiently large.

The encoding algorithm $\enc$ will be very simple. First, we let $Q = \min (s(n)^{0.1}, 2^{\sqrt{\log n}})$. Then, we take a code $\LDC$ as in \cref{thm:ldc-bin} with parameters $n, \eps, Q$, and let $y = \LDC(x)$ have length $N=n \cdot (\log_Q n)^{O(\log_Q n)}$. Then, we simply define $\enc$ as follows:

\begin{defn}
Define $\enc(x) = \LDC(x)^k = y^k$, where $y^k$ denotes the string $y$ repeated $k$ times, and $k = Q^2n/s$.
\end{defn}

Note that we will then have $m \coloneq |\enc(x)| = kN$, so the adversary will be permitted at most $(1/4 - 10\e)kN$ errors.)

\begin{algorithm}[t] 
\caption{Decoding algorithm $\dec$ for \cref{thm:n2const}}
\label{alg:init}
\begin{algorithmic}[1]
\State Let $Q = \min(s^{0.1}, 2^{\sqrt{\log n}})$, $k = k = Q^2n/s$.
\State $\e \gets \e/10$.
% \State Let $\LDC$ be a code as given in \cref{thm:ldc-bin} with parameters $n, \eps, Q$.
\State Sample $j_1, \dots, j_a \sim \Dd$, where $\Dd$ is as in \cref{thm:ldc-bin}.
\State Sample $j_{u+1}, \dots, j_{u+v} \in [N]$ uniformly randomly.
\State For each $1 \le t \le u+v$, set $P_0^t, P_1^t \ot 0$.
\While{there exists $t$ such that $P_0^t, P_1^t \le (1-\e)k/2$}
    \State Read next copy $y^{(i)}$, and perform smooth local decoding on $y^{(i)}$ as in \cref{thm:ldc-bin}, $u+v$ times in parallel, to get confidences $p_0^t, p_1^t$ for each index $j_t$.
    \State For each $t$, set $P_0^t \gets P_0^t + p_0^t$ and $P_1^t \gets P_1^t + p_1^t$.
\EndWhile
\State For each $t$, let $b_t$ be such that $P_{b_t}^t > (1-\e)k/2$.
\While{there are bits remaining to output}
    \State Read the next copy $y^{(i)}$, keeping a counter $c$ tracking how many $u < t \le u+v$ satisfy $y^{(i)}_t = b_t$. 
    \Comment{We will prove that the copies will not run out before decoding is complete.}
    \State In parallel with the above, perform local list decoding with advice as in \cref{thm:ldc-bin} using $y_{j_t} = b_t$ for $1 \le t \le u$, in parallel for the next $r = s(n)/Q^2$ bits that have not yet been output.
    \If{$c < (1/2 - \e)v$}
        \State Output the decoded bits.
    \EndIf
\EndWhile
\end{algorithmic}
\label{alg:dec-n2const}
\end{algorithm} 

We now describe how the decoding algorithm $\dec$ will work, as given formally in \cref{alg:dec-n2const}. Let $u$ be chosen as in the ``local decoding with advice'' property in \cref{thm:ldc-bin}, and let $v=(\log n)^2$. Pick indices $j_1, \dots, j_{u+v} \in [1, N]$, where $j_1, \dots, j_u$ are chosen according to the distribution $\Dd$ in \cref{thm:ldc-bin} and $j_{u+1}, \dots, j_{u+v}$ are chosen uniformly randomly. 

Informally, these bits will be used to create a ``checksum'' for $y = \LDC(x)$. More explicitly, we will use the first part of the stream to determine what $y_{j_1}, \dots, y_{j_{u+v}}$ are with high probability. Then, for the remaining (corrupted) copies of $y$, we will first check whether they have at most $1/4-\e$ errors by comparing their bits at $j_{u+1}, \dots, j_{u+v}$. Then, if they do, we will use the local decoding with advice property in \cref{thm:ldc-bin} to recover bits of $x$.

More specifically, for each index $j_t$, we will keep track of quantities $P_0^t, P_1^t$, which are the total confidence that $y_{j_t}$ is 0 or 1, respecitively. When receiving the stream, we will perform smooth local decoding (as described in \cref{thm:ldc-bin}) on each (corrupted) copy of $y$, to obtain confidences $p_0, p_1$ for each index $j_t$ (in parallel for each $t$). For each $t$, we then increment $P_0^t, P_1^t$, respectively, by the obtained $p_0, p_1$. We show the following claim:

\begin{claim} \label{claim:err-count}
With high probability, the following holds for all $t$ and at every step of the algorithm: Let $b = y_{j_t}$. Suppose that, after reading $\ell$ corrupted copies of $y$, we have $P_b^t \le (1-\e)k/2$. Then, the number of errors in the first $\ell$ copies of $y$ is at least $\f 12 (1-\eps)(\ell - \f 12 k)N$.
\end{claim}
\begin{proof}
We prove the statement for a particular choice of $t, \ell$, since a union bound will imply the statement simultaneously for all choices.

Suppose that the total number of errors in the first $\ell$ copies of $y$ is $xN$. Then, by \cref{thm:ldc-bin}, with high probability we have for each $i$ that $p_b^{(i)} \ge 1 - 2\delta(y^{(i)}, y) - \eps$. Therefore, we have
\begin{align*}
P_b^t &= \s_{i=1}^\ell p^{(i)} \\
&\ge (1-\e)\ell - 2\s_{i-1}^\ell \delta(y^{(i)}, y) \\
&= (1-\e)\ell - 2x.
\end{align*}
On the other hand, by assumption, $P_b^t \le (1-\e)k/2$. Therefore, we have $(1-\e)\ell - 2x \le (1-\e)k/2$. Rearranging, we obtain $x \le \f 12 (1-\eps)(\ell - \f 12 k)$, as desired.
\end{proof}

Note that this implies that if for $P_b^t \le (1-\e)k/2$ for $b = y_{j_t}$, $\ell = (1-\e)k$, and any $t$, then the number of errors is at least $\f 12 (1-\eps)((1/2 - \e) k)N > (1/4-\e)kN$, a contradiction. Thus, when $\ell = (1-\e)k$, we have $P_b^t > (1-\e)k/2$, and since $P_0^t + P_1^t = \ell$, we also have $P_{1-b}^t < (1-\e)k/2$. Thus, $b=y_{j_t}$ is the unique $b \in \{0, 1\}$ such that $P_b^t$ exceeds $(1-\e)k/2$ before $P_{1-b}^t$.

Thus, we let $\ell$ be the smallest value such that after reading $\ell$ corrupted copies of $y$, we have for every $t$ that there exists some $b_t$ such that $P_b^t > (1-\e)k/2$. By the previous paragraph, we have (with high probability) $b_t = y_{j_t}$ for every $t$ and $\ell \le (1-\e)k$. We store all these values $y_{j_t}$. Moreover, by the minimality of $\ell$, by \cref{claim:err-count}, the number of errors in the first $\ell-1$ copies of $y$ is at least $\f 12 (1-\eps)(\ell - 1 - \f 12 k)N$. Thus, the total number of errors remaining after the $\ell$-th copy of $y$ is at most
\begin{align*}
\p{\f 14 - 10\e}kN - \f 12 (1-\eps)\p{\ell - 1 - \f 12 k}N  
&\le   \p{\f 12 - 3\e} (k - \ell)N - 5\e kN.
\end{align*}

Now, for all copies $y^{(i)}$ for $i > \ell$, we do the following two things in parallel. First, we count how many $u < t \le u+v$ satisfy $y^{(i)}_{j_t} \ne y_{j_t}$; let $c$ be this count of incorrect bits. Second, for the next $r = s(n)/Q^2$ bits which we haven't yet output, we run the ``local list decoding with advice'' algorithm in \cref{thm:ldc-bin} (in parallel for all $r$ bits). Then, if $c < (1/2-2\e)v$, we output the $r$ resulting decoded bits; otherwise, we do nothing and just continue to the next copy $y^{(i-1)}$.

Since $j_{u+1}, \dots, j_{u+v}$ are chosen uniformly randomly and independently from $y^{(i)}$, we can apply a Chernoff bound on $c$. Since $v = (\log n)^2$, we have with high probability that if we did output the bits, then the number of errors was at most $(1/2-\e)N$, so by \cref{thm:ldc-bin}, all bits that were output are correct.

Moreover, we also have by the same Chernoff bound that (with high probability) if we did not output any bits in step $i$, then the number of errors in $y^{(i)}$ must have been at least $(1/2-3\e)N$. Recall that the total number of errors remaining in the last $k-\ell$ copies of $y$ was at most $(1/2-3\e)(k-\ell)N - 5\e kN$. Thus, the number of copies with less than $(1/2-3\e)N$ errors is at least
\[\f{5\e kN}{(1/2-3\e)N} > 10k \ge \f nr.\]
Thus, in the end, we will have output at least $(n/r) \cdot r \ge n$ bits with high probability, so we are done.

Now we analyze the space complexity of the algorithm. During the first part, we repeatedly did smooth local decoding $u+v = O(Q + (\log n)^2)$ times in parallel, which takes $O(Q^2)$ space for each bit by \cref{thm:ldc-bin}. Thus, the total space complexity of this first step is $O((Q + (\log n)^2)Q^2)$, which is less than $s(n)$ since $s(n) > (\log n)^C$ and $Q \le s(n)^0.1$. %Also, storing the indices $j_t$ and counts $P_b^t$ takes at most $O(\log n)$ bits for each $t$, since the $j_t$ are integers which are at most $N < \poly(n)$ (we will check this soon) and the the $P_b^t$ are multiples of $1/(4Q)$ which are at most $n$. Thus the total number of bits to store these is $O((Q + (\log n)^2)\log n)$ which is also at most $s(n)$.
In the second half of the algorithm we performed local list decoding with advice $r = s(n)/Q^2$ times in parallel, and each instance takes $O(Q^2)$ space per \cref{thm:ldc-bin}, so the total space for this part of the algorithm is $O(s(n))$. Thus, the decoding algorithm takes $O(s(n))$ space in total, as desired.

Finally, we check that the length of the encoding $m = kN$ is the desired quantity. Recall that $N = n \cdot (\log_Q n)^{O(\log_Q n)}$, where $Q = \min(s(n)^{0.1}, 2^{\sqrt{\log n}})$. Now we check the two regimes of \cref{thm:n2const}:
\begin{itemize}
    \item If $s(n) = \Omega((\log n)^t)$, then $Q \ge (\log n)^{0.1t}$, so $N \le n \cdot (\log n)^{c \log n / t \log \log n} = n^{1 + c/t}$, for some constant $c$. Also, $k = Q^2 n / s(n) \le n$. Thus, $m = kN = n^{2 + O(1/t)}$.
    \item If $s(n) = (\log n)^{\omega(1)}$, then we also have $Q = (\log n)^{\omega(1)}$, so $N \le n \cdot (\log n)^{\log n / \omega(\log \log n)} = n^{1+o(1)}$. Also, since $Q \le 2^{\sqrt{\log n}} = n^{o(1)}$, we have $k = Q^2 n / s(n) = n^{1+o(1)}/s(n)$. Thus, $m = kN = n^{2 + o(1)}/s(n)$.
\end{itemize}
This completes the proof of \cref{thm:n2const}.

\section{Stream decodable codes require quadratic length} \label{sec:lower}

\renewcommand{\eps}{\rho}

We now prove our lower bound, \cref{thm:n2lower} (restated below), demonstrating that the construction in Section~\ref{sec:n2-const} is essentially tight. That is, any error-correcting code that can be decoded with failure probability at most $1/2n^2$ by a stream permitting $s(n)$ space must have encoding length at least $\Omega\left( \frac{n^2}{s(n)} \right)$. 

\quadlower*

\begin{proof} 
Suppose otherwise; that is, suppose that there is a $(\e, m, s)$-coding scheme for streams, where $m = \e n^2 / 10^4 s$, $\e$ is a fixed constant, and $s = s(n) \ge \log n$ (and $n$ is sufficiently large). Also, we may assume that $s < n / 100$ (otherwise the statement is obvious).

We will then demonstrate how to construct an adversarial input for this coding scheme, so that $\dec$ fails with probability at least $1/2n^2$. First, note that we can assume that $\dec$ does not output anything when it receives a 0 bit (except at the end of the stream): instead, we may have $\dec$ keep track of the length of the current run of 0s (using only $O(\log n) \lesssim s$ memory), and process all the 0s when it encounters the next 1. In particular, we will have the adversary replace several parts of the input with all 0s, and thus we can assume that the algorithm does not output anything at these parts (except perhaps the last block).

For an input string $x \in \{0, 1\}^n$, the encoding $\enc(x)$ then has length $m$. We split $\enc(x)$ up into $k = n / 100s$ contiguous ``blocks'' which each consist of $\ell = \e n / 100$ bits; denote these blocks $B_1(x), \dots, B_k(x)$ (we will sometimes abbreviate $B_i(x)$ by just $B_i$). We will consider what the algorithm $\dec$ may output during each block, assuming that the block $B_i$ is uncorrupted. Essentially, we will show that there is a fixed set of roughly $\ell$ indices such that $\dec$ essentially only outputs indices from this set while it is processing block $i$.

To this end, we will let $a_i \in \{0, 1\}^s$ denote the contents of the memory of $\dec$ right before receiving block $i$. Note that $a_i$ is random and may depend on the randomness of $\dec$, as well as on the bits that the adversary has changed in previous blocks. We will mostly restrict our attention to values of $a_i$ that do not cause the algorithm to fail with significant probability. Specifically, we say that $a_i$ is \textit{good with respect to $x$}, or just \textit{good} (for particular values of $x$ and $i$), if the probability that $\dec$ outputs an incorrect bit during block $i$ with starting memory $a_i$ is at most $1/n^2$. (This probability is taken over only the randomness of the algorithm $\dec$, since the contents of block $B_i$ are a deterministic function of $x$.)

Now, suppose that the decoder $\dec$ currently has memory state $a_i$ and is about to receive block $i$ (whose contents are $B_i$). While it processes $B_i$, it will output various bits of $x$; that is, there are various pairs $(j, b)$ for which $\dec$ will output that $x_j=b$. (We assume, as we may, that $\dec$ keeps track of which index it is on, so we can determine $j$ from the memory contents of $\dec$.) When it does so, we say that $\dec$ outputs the pair $(j, b)$. Then, let $T(a_i, B_i)$ be the set of all $(j, b)$ which $\dec$ outputs with probability at least $1/n^2$ when it receives $B_i$ with initial memory contents $a_i$. (Again, this probability is only over the randomness of $\dec$.) Note that if $a_i$ is good (with respect to $x$), then $T(a_i, B_i)$ must match $x$ (that is, $x_j = b$ for all $(j, b) \in T(a_i, B_i)$). Note that $T(a_i, B_i)$ is a \textit{deterministic} function of $a_i, B_i$.

% Also, let $T_i = T(B_i, a_i)$ be the set of indices that $\dec$ outputs during block $i$, along with the predicted value of $x$ at those indices. That is, the elements of $T_i$ are the tuples $(j, b)$ for which $\dec$ outputs that $x_j=b$; note that the values of $j$ must be contiguous. (Observe that $T_i$ is a randomized function of $B_i$ and $a_i$, and does not depend on the rest of $\enc(x)$.\msmargincomment{technically in order for this to be true, the algorithm must keep track of what bit it's currently on, which we can assume, but this section is already filled with disclaimers so i'm going to leave it out for now}) %Sometimes, we will abuse notation by letting $T_i$ consist of just the indices $j$ instead of the tuples $(j, b)$.

We are now ready to prove the following lemma.

\begin{lemma} \label{lemma:lb-1}
% Suppose that the input string $x$ is chosen uniformly at random. Then, with probability $1-o(1)$, for all $i$, the following holds: There exists a set of indices $S_i$ of size $n^{1-\rho+o(1)}$ (depending only on $x$ and $i$) such that for all good $a_i$, with probability at least $1-1/n^5$ (over the randomness of $T$), $|T(B_i(x), a_i) \setminus S_i| \le s(n)$. \todo{change upper bound to whatever it actually is}
There exists $x \in \{0, 1\}^n$ such that the following holds for all $i$: There is a set $S_i$ of size at most $3\ell$ %(depending only on $x$ and $i$) 
such that for all good $a_i$, we have $|T(a_i, B_i(x)) \setminus S_i| < 3s$.
\end{lemma}
\begin{proof}
% It suffices to fix $i$ and show that (for uniformly random $x$) the statement holds with probability $1-O(1/n)$. We will show how to construct $S_i$, allowing the process to ``fail" with probability $O(1/n)$.

% Let $A$ be the set of all good $a_i$; we will construct $S_i$ iteratively by repeatedly adding indices to it, while removing from $A$ the $a_i$ for which the property has already been satisfied.

Let $a_i^{(1)}, \dots, a_i^{(r)}$ each be good $a_i$ (with respect to a particular choice of $x$ and $i$), where $r=\ell / s$. Consider the following union: \[\Tt = \bigcup_{1 \le j \le r} T(a_i^{(j)}, B_i).\]
Essentially, if we can show, for a particular $x$, that this union is always small, we will then be able to show that $T(a_i, B_i)$ cannot take too large a range of values as (good) $a_i$ varies, because the union of any $r$ such instances will have small size. To this end, we first show the following claim.
\begin{claim} \label{claim:T-small}
There exists $x$ such that, for every $i$, the union $\Tt$ always has size at most $3\ell$ (no matter the choice of $r$ good $a_i^{(j)}$'s).
\end{claim}
\begin{proof}
First observe that since each $T(a_i, B_i)$ must match $x$, this means that $\Tt$ must also match $x$ (recall that this means that for every $(j, b) \in \Tt$, we have $x_j=b$). However, $\Tt$ is a deterministic function of $(B_i, a_i^{(1)}, \dots, a_i^{(r)})$, which consists of $2\ell$ bits. Therefore, there are only at most $2^{2\ell}$ possible values that $\Tt$ can take for any particular $i$. Thus, in total (over all $i$) there are at most $n \cdot 2^{2\ell} < 2^{3\ell}$ possible values for $\Tt$.

However, each possible value of $\Tt$ that has size at least $3\ell$ can only match a $2^{-3\ell}$ proportion of $x \in \{0, 1\}^n$. Therefore, there exists some $x$ which does not match any possible $\Tt$ of size at least $3\ell$, thus proving the claim.
\end{proof}

Now fix $x$ such that \cref{claim:T-small} holds, and let $i$ be arbitrary. We will now finish the proof of \cref{lemma:lb-1} by constructing $S_i$. Let $\Ff$ be the family that consists of $T(a_i, B_i)$ for all good $a_i$. \cref{claim:T-small} means that the union of any $r$ sets in $\Ff$ has size at most $3\ell$. We wish to find $S_i$ which contains all but at most $3s$ elements of each $T \in \Ff$.  

Now, construct $S_i$ in steps as follows: at each step, find $T \in \Ff$ which has more than $3s$ elements which are not in $S_i$, and add all its elements to $S_i$. This process terminates when there is no such $T$ remaining. Obviously this set satisfies $|T \setminus S_i| \le 3s$ for all $T \in \Ff$, so it remains only to check that $|S_i| < 3\ell$. Indeed, suppose that $|S_i| \ge 3\ell$; consider the first step in which its size reached or exceeded $3\ell$. Note that at each step the size of $S_i$ increases by more than $3s$, so in total the number of steps for $|S_i|$ to reach $3\ell$ is at most $\f{3\ell}{3s} = r$. But then $S_i$ is the union of at most $r$ sets in $\Ff$ and has size at least $3\ell$, contradicting \cref{claim:T-small}. Thus, $S_i$ has the desired properties, completing the proof of \cref{lemma:lb-1}. 
\end{proof}

With this lemma proven, we return to the proof of \cref{thm:n2lower}. We will now demonstrate a strategy for the adversary such that, with probability at least $1/2n^2$, the decoding algorithm $\dec$ fails to output $x$. Fix the input $x$ and sets $S_i$ such that \cref{lemma:lb-1} is satisfied.

Now, the adversary picks a uniformly random index $j \in \{1, \dots, n/2\}$. Then, for each $i$ such that $S_i$ contains at least $5s(n)$ indices in $[j + 10(i-1)s, j+10is)$, the adversary replaces the whole block $B_i$ with zeros (unless it is the last block). We will first show that $\dec$ must fail on this input with probability at least $1/2n^2$. Let us suppose otherwise. 

As before, let $a_i$ be the (random) memory state of $\dec$ before processing $B_i$. Note that with probability at least $1/2$, all the $a_i$ are good (since in the cases where they are not, $\dec$ fails with probability at least $1/n^2$). If they are all good, then by the definition of $T$ and a union bound (over the block number $i$ and the indices $j$), with probability at least $0.99$, at every block $B_i$, the indices output during block $i$ are all in $T(a_i, B_i)$.  In this case, observe that during block $i$, $\dec$ may never output the index $j+10is$ (or any greater index). Indeed, if this were not the case, during some block $\dec$ would have to output everything in $[j + 10(i-1)s, j+10is)$, but then we would have $|T(a_i, B_i) \setminus S_i| > 5s$, contradicting \cref{lemma:lb-1}.

Thus, right before the last block, the algorithm cannot have output any index past $j + 10ks = j + n / 10$. Then, in the last block, the algorithm outputs at most $|S_i|+3s \le 3\ell + 3s$ by \cref{lemma:lb-1}. Therefore, overall, with probability at least $0.49$, the algorithm outputs at most $j + n/10 + 3\ell + 3s < n$ indices, and thus does not output all of $x$.

Therefore we have shown that, under this strategy for the adversary, the algorithm must fail on $x$ with probability at least $1/2n^2$. It remains only to show that the adversary deletes (i.e., replaces with $0$'s) at most an $\eps$ fraction of blocks. Indeed, it is enough to show that at most an $\eps$ fraction of blocks are deleted in expectation, since the adversary can pick $j$ such that the fewest blocks are deleted. Fix a block $B_i$. The probability that $B_i$ gets deleted is equal to the probability that $S_i$ has at least $5s$ indices in the interval $[j + 10(i-1)s, j+10is)$. For any fixed $j' \in B_i$, there are at most $10s$ choices of $j$ such that $j'$ lands in this interval, so the probability that it does is at most $20s/n$ (since $j$ is chosen uniformly at random from $n/2$ choices). Thus the expected number of indices in $S_i$ in the interval is $(20s/n)|S_i| \le 60s\ell/n$. By Markov's inequality, the probability that this is at least $5s$ is at most $12\ell/n < \eps$. Therefore, the probability that $B_i$ is replaced with $0$'s is at most $\eps$ for each block $B_i$. The expected number of blocks replaced by $0$'s is therefore at most $\eps k$.

Putting everything together, the adversary has a strategy which deletes at most an $\eps$ fraction of blocks (and thus at most an $\eps$ fraction of the bits) which causes $\dec$ to fail with probability at least $1/2n^2$. This completes the proof of \cref{thm:n2lower}.
\end{proof}
\renewcommand{\eps}{\varepsilon}
\section{Stream decodable codes for linear functions of near linear length} \label{sec:linear}

Our final result is a noise-resilient encoding of essentially linear length that admits efficient stream decoding of arbitrary linear functions. The family of linear functions is defined as the functions $f:\{0,1\}^{n} \to \{0,1\}$ for which there exists $y\in \{0,1\}^n$ such that $f(x) = x\cdot y \mod 2$.

\linearconst*

\paragraph{Parameters and notation.} Throughout this section, fix $\eps$ (we will hide dependence on $\e$ in big $O$ notation), $\delta$ if it exists, and the space function $s$. Let $n$ represent the length of Alice's message. 
%Let $C'$ be the absolute constant from Theorem~\ref{thm:ldc-large}. We assume that $s(n)>(\log n)^{C:=10C'}$. Set the following parameters:
We assume that $s(n)>(\log n)^{1000}$. Set the following parameters:
\[
    r = s(n)^{0.2} \quad \text{and} \quad d=\frac{\log n}{\log r} \quad \text{and} \quad \eps'=\frac{\eps}{10d} ~.
\]
    Let $\bbK$ be $\bbF_{2^k}$ where $\e'^{-10} < 2^k \le 2\e'^{-10}$ so that the condition of Theorem~\ref{thm:ldc-large} is satisfied for $\bbK$ and $\e'$.

%Let $\LDC : \bbK^r \to \bbK^R$ with be a linear locally decodable code with locality $Q>\left( \eps'^{-1}\log n \right)^C$ (chosen later) satisfying the guarantees of Theorem~\ref{thm:ldc-large} for $\eps'$. We remark that we will also choose $R$ later, potentially making it larger than the bound given by Theorem~\ref{thm:ldc-large}. The required code will still exist, since we can simply take a code given by padding the length $r$ input -- it is only required that $Q$ still satisfies the constraint in the theorem, which it will since $Q>\left( \eps'^{-1}\log n \right)^C$.

Let $\LDC : \bbK^r \to \bbK^R$ with be a linear locally decodable code satisfying the guarantees of Theorem~\ref{thm:ldc-large} for $\eps'$. The locality is $Q>\left( \eps'^{-1}\log r \right)^{100}$. This is satisfied whenever $Q>\left( d\log r\right)^{1000}$ because 
\[
    \left( \eps'^{-1}\log r \right)^{100} \leq \left( d\eps^{-1}\log r \right)^{100} \leq \left( d (\log r)^2 \right)^{100} \leq \left( d\log r\right)^{1000}
\]
since $\log r$ is sufficiently large compared to $\eps^{-1}$. 
We will set $Q$ subject to this constraint later. Also note that $Q>\eps^{-1}\log n = \eps^{-1}d\log r$ which is a fact we will use later and $Q<r$ must be satisfied. 

This gives us a value of $R\geq (\eps'^{-1}\log_Q r)^{100\log_Q r}$ which is satisfied if $R\geq (d\log_Q r)^{150\log_Q r}$. We will actually set $R$ later, subject to this constraint. We can make $R$ larger as needed by a variety of methods, such as padding the input or duplicating each bit of the code.

We assume for simplicity that $r$ and $d$ are integers. It will be useful to index Alice's (the sender's) input $x\in \{0,1\}^n$ by a tuple in $[r]^d$ rather than an integer in $[n]$. Whenever we say an event occurs with high probability, we mean with high probability in $n$ unless specified otherwise. We refer the reader to Section~\ref{sec:prelims} to review notation used in this section.

\subsection{Statement of encoding scheme}

The encoding $\enc(x)$ that Alice (the encoder) sends in the stream is a tensor code. To this end, we begin by defining a tensor power of a linear code $\C$. We remark that tensor products of distinct linear codes can also be defined, but we will only need to take a tensor power of one code.

\begin{definition}[$\C^{\otimes k}$] \label{def:tensor}
    Let $\C: \bbK^m \to \bbK^{M}$ be a linear error correcting code on strings of length $N\in \bbN$ on some alphabet $\bbK$. Since the code is linear, $\C$ is an $M\times m$ matrix. Then the $k$-th tensor power of this encoding matrix $\C^{\otimes k}: \{0,1\}^{[m]^k} \to \{0,1\}^{[M]^k}$ is the encoding function $\C^{\otimes k}$. We note that for any code $\C$, it holds that $\C^{\otimes 0} : \bbK \to \bbK$ is the identity function.
\end{definition}

Next, we will state the encoding $\enc(x)$ that Alice (the encoder) sends in the stream.

\begin{definition}[$\C_{inner}$] \label{def:c-inner}
    Let $\C_{inner} : \bbK \to \{0,1\}^{O(1)}$ be a distance $(1-\eps/4)$ linear code guaranteed by Theorem~\ref{thm:binary-ecc}. It's length is $N_{inner}$.
\end{definition}

\begin{definition}[$\enc(x)$] \label{def:enc-linear}
    Alice's encoding is $\enc(x)$ is defined as follows. Viewing $x$ as an element of $\bbK^n$ (which we may since $\bbK$ is of the form $\bbF_{2^k}$, she computes $\LDC^{\otimes d}(x)$ (where the message and codeword bits are both taken in lexicographic order) and concatenates this with $\C_{inner}$.
\end{definition}

We note that Alice's encoding is length $R^d$, and her encoding takes time at most $\poly(n)$. We'll later choose our parameters to satisfy the conditions of Theorem~\ref{thm:linear-const}.

\subsection{Statement of decoding scheme}

Let Bob's private vector be $\ell = \langle \ell_{(1,\ldots,1)}, \ldots,\ell_{(r,\ldots,r)} \rangle$, (here the ordering is lexicographic). The function Bob is trying to compute is $\ell \cdot x=\ell_{(1,\ldots,1)}x_{(1,\ldots,1)} + \ldots + \ell_{(r,\ldots,r)}x_{(r,\ldots,r)}$. For $a\leq d$, the vector $\ell_{(i_1\ldots i_{d-a})}$ is defined to be an $r^a$ dimensional vector that denotes $\ell$ restricted to indices where the first $d-a$ entries of the tuple are $(i_1\ldots i_{d-a})$. Throughout, we will view $\ell$ and $x$ as elements of $\bbK$ rather than $\bbF_2$ and note that it suffices to compute $\ell\cdot x$ in $\bbK$.
%\vgnote{Better to use $i_1,\dots,i_{d-a}$ and define associated $r^a$ dimensional vector as that's how notation is used later.} 
The same notation applies for $x$ or any other string canonically indexed by tuples.

%We'll state the algorithm in two parts. First, we'll describe an algorithm that comprises the recursive step. This algorithm solves the following problem. Given a stream 

Before we state our main algorithm, we efficiently construct lists $\qlist_1, \ldots, \qlist_r$ satisfying certain properties that Bob can find in $s(n)$ space and $\poly(n)$ time. These lists will be the indices of LDC that we query for each $i$, and it will be important that they overlap as little as possible.

\begin{lemma} \label{lem:qlist}
    Given a locally decodable code $\LDC: \bbK^m \to \bbK^M$ satisfying $m>\log n$ with $q$ queries satisfying the requirements of Theorem~\ref{thm:ldc-large} for $\eps'$, there is a (randomized) algorithm (permitted $\frac{1}{m^{\omega(1)}}$ probability of failure) that generates lists $\qlist_1, \ldots, \qlist_m$ in time $\poly(mq)$ and space $O(mq^2)$ such that the following holds. The smooth local decoding algorithm for each index $i$ queries only the indices $\qlist_i$, and for each  $z\in\{0,1\}^M$, with high probability in $m$, satisfies the smoothness guarantees in Theorem~\ref{thm:ldc-large} for all $i$. Moreover, no $I \in [M]$ appears in more than $\left\lceil\frac{3mq^2}{M}\right\rceil$ lists.
\end{lemma}

\begin{proof}
    Generate lists $\qlist_i$ one-by-one according to the local decoding algorithm in Theorem~\ref{thm:ldc-large}. An apparent contradiction occurs if a given $I\in [M]$ appears for the $\left( \left\lceil\frac{3mq^2}{M}\right\rceil +1\right)$-th time. In this case, we re-sample up to $q^2$ times, until the contradiction no longer occurs. The time guarantee is satisfied because the local decoding algorithm takes $\poly(mq)$ time, and we are only resampling up to $\poly(q)$ times. The space guarantee is satisfied because only $m$ lists of length $q$ are being stored at any time. The elements of each list are indices of length 
    \[
        \log M = \log \left( m \cdot (\log_q m)^{O(\log_q m)} \right) O(\log (m)^2) < O(q).
    \]
    This is true because $q > \log(m)^4$, which is a constraint already present in Theorem~\ref{thm:ldc-large}.
    Along with this, we need a counter of length at most $\log(q^2)$. Also, generating the candidate lists requires significantly less than $O(mq)$ space. 
    
    We now show that with high probability in $m$, the resampling process will terminate within $q^2$ steps. In other words, all chosen indices will have appeared at most $\left\lceil\frac{3mq^2}{M}\right\rceil$ times so far. Say an index is \emph{bad} if it has appeared at least $\left\lceil\frac{3mq^2}{M}\right\rceil$ times so far. At any point, there are at most $\frac{qm}{\left\lceil\frac{3mq^2}{M}\right\rceil}<\frac{M}{2q}$ indices that are bad. Since in our choosing process, every index has at most a $\frac{1.1q}{M}$ chance of getting chosen by Theorem~\ref{thm:ldc-large}, on average $0.6$ bad indices are chosen. By Markov's inequality, this means that the probability $0$ bad indices are chosen is at least $0.4$. Iterating $q$ times gives a success probability of at least $1-2^{-q}=1-m^{-\omega(1)}$ of finding a subset that doesn't violate the condition that each $I$ appears infrequently.

    Now, we need to verify that all the decoding algorithms will output a valid pair of probabilities $p(0),p(1)$ with high probability in $m$ for a given string $z$. This is true because we only generated $m$ independent decoding algorithms, and Theorem~\ref{thm:ldc-large} guaranteed a $1-\frac{1}{m^{\omega(1)}}$ success probability of each algorithm, so by union bound, the success probability is still $1-\frac{1}{m^{\omega(1)}}$.
\end{proof}

We are now ready to state the main algorithm. We refer the reader to Section~\ref{sec:overview} for an informal description of the algorithm.

\begin{algorithm}[H]
\caption{Bob's decoding algorithm $\mathsf{linear\_dec}$}
\label{alg:linear}
\renewcommand{\thealgorithm}{}
\floatname{algorithm}{}

\begin{algorithmic}[1]

\State \textbf{input:} $n, d\in \bbN$, $\LDC : \{0,1\}^r \to \{0,1\}^R$, and stream $z \in \{ 0, 1 \}^{R^d}$. 
\State

\Function{$\mathsf{recurse\_linear}$}{integer $a \leq d$, stream $z \in \{0,1\}^{N_{inner}R^a}$, tuple $T\in [r]^{d-a}$}: \Comment{At each recursive layer, want to compute $\ell_T\cdot x$.}

    \If {$a=0$}
        \State Set $\sigma$ to the decoding of $z$ via $\C_{inner}$ with probability $1-2\delta(z,\C_{inner}(\sigma))$ and otherwise to $\bot$. This is the only step where the stream is read, and it happens in unison for all instances reading this stream; the randomness in deciding $\sigma$ is the same for all of them. \label{line:C-inner}
        \State \Return $\ell_T\cdot \sigma$ ($\bot$ if $\sigma=\bot$).
    \EndIf
    \State Pick lists $\qlist_1, \ldots, \qlist_r$ as in Lemma~\ref{lem:qlist}. \Comment{These are the indices of $\LDC$ queried by a local decoding algorithm for each message index, and no index appears in more than $\left\lceil\frac{3rQ^2}{R}\right\rceil$ lists.}
    \State Define maps $\cM_i: \qlist_i \to \bbK$ for $i\in [r]$ as follows:  \Comment{Outputs of the queries $\qlist_i$ made to a corrupted version of $\LDC \left(\ell_{T|i}\cdot x_1 \ldots \ell_{T|i}\cdot x_r \right)$ (computed recursively), used to decode $\ell_{T|i}\cdot x_i$.} 
    \For {$I \in [R]$}
        \State Let $z_I$ denote the next $R^{a-1}$ elements of the stream $z$. \Comment{Just notation, this step does not read any of $z$.}
        \For {$i$ such that $I \in \qlist_i$}
            \State Set $\cM_i[I] \gets \mathsf{recurse\_linear}(a-1, z_I, T|i)$. \Comment{Run all in parallel, since they all use the same stream $z_I$.}
        \EndFor
    \EndFor
    \For {$i \in [r]$}
        \State Compute $p_i(\sigma)$ for all $\sigma\in \bbK$ using local decoding on index $i$ with queries $\cM_i$. 
        \State Set the guess $\sigma_i = \argmax_{\sigma}p_i(\sigma) \in \bbK$ and likelihood $p_i=\max_\sigma p_i(\sigma)$.
        %\State If there exists a (unique) $\sigma\in \bbK$ with $p(\sigma)>0.5$, then let $\sigma_i=\sigma$ with confidence $c_i=2p(\sigma)-1$. Otherwise, let $\sigma_i=\bot$ with confidence $c_i=0$.
    \EndFor
    \State Let $p=\min_i p_i$.
    \State \Return $\sigma_1+\ldots+\sigma_r$ with probability $p$ and otherwise $\bot$.

\EndFunction
\State

\State Compute $\mathsf{recurse\_linear}(d, z, ())$ in parallel $(\log n)^2$ times. When the output is $\bot$, assign it $0$ or $1$ randomly and otherwise project to $\bbF_2$. Let majority output be $\sigma$. \Comment{$\mathsf{recurse\_linear}(d, z, ())$ is a guess for $\ell\cdot x$.}
\State \textbf{output:} $\sigma$.

\end{algorithmic}
\end{algorithm}

\subsection{Proof of Theorem~\ref{thm:linear-const}}

First, we'll show that the protocol runs in time $\poly(n)$. 

\begin{lemma} \label{lem:linear-time}
    Bob's decoding algorithm $\mathsf{linear\_dec}$ runs in time $\poly(n)$.
\end{lemma}

\begin{proof}
    We'll show by induction that each iteration of $\mathsf{recurse\_linear}(a,z,T)$ takes time $O(R^{2a}\cdot \polylog n)$. 

    This holds if $a=0$: Bob only needs to process a tuple $t$ of length $d<\polylog n$ and index into $\ell$. This takes time $d\log r< \polylog n$. Then, he only has to compute the product of two elements of $\bbK$ which also takes time less than $\polylog(n)$.

    Next, we establish the inductive step. First, Bob computes at most $rR<R^2$ iterations of $\mathsf{recurse\_linear}(a-1,z_I,T|i)$. Each takes time $O(R^{2(a-1)}\cdot \polylog n)$, so in total this takes $O(R^{2a}\cdot \polylog n)$. Next, Bob computes each $b_i$ and confidence $c_i$. This takes time at most $\polylog n$ for each $i$ from the local decoding algorithm, so time $r\cdot \polylog n<R\cdot \polylog n$ in total. Adding them and returning a guess according to the confidence takes time at most $R\cdot \polylog n$ as well.
\end{proof}

Next, we show the space guarantee of the protocol.

\begin{lemma} \label{lem:linear-space}
    Consider a call to the function $\mathsf{recurse\_linear}(a,z,T)$. The amount of space used by the algorithm is $O\left(a \left\lceil\frac{3rQ^2}{R}\right\rceil^a \cdot RQ^2\right)$.
\end{lemma}

\begin{proof}
    Throughout, we'll implicitly use the bound that $Q>\eps^{-1}\log n$.
    
    We'll show this by induction on $a$; let $f(a)$ be the inductive bound on space for $a$. For the base case of $a=0$, the function accesses $\ell_T$, which takes $d \log r \le O(RQ)$ space to write down $T$. It then calculates the product of two elements of $\bbK$, which takes $O(\log |\bbK|)<O(10\log(\eps^{-1}d))< O(Q)$ space. Thus $f(0) = O(Q) \le O(RQ^2)$, as desired.

    For the inductive step, the first piece of the algorithm is to generate the lists $\qlist_i$. We showed in Lemma~\ref{lem:qlist} that this takes $O(rQ^2)<O(RQ^2)$ space. Then, we store the maps $\cM_i$. Each of these takes space $O(Q(\log r+\log |\bbK|))< O(Q^2)$, for a total of $O(RQ^2)$ space. Next, for each $I\in [R]$, we'll run a recursive calculation on the stream $z_I$. For a fixed $I$, we have multiple recursive calls that run in parallel. By the guarantees of our particular $\LDC$, each index appears in at most $\left\lceil\frac{3rQ^2}{R}\right\rceil$ lists, and so the number of calls occurring in parallel is at most $\left\lceil\frac{3rQ^2}{R}\right\rceil$. The final calculation of $\sigma_1,\ldots, \sigma_r$ also requires $O(RQ^2)$ space. Therefore, in total, the space usage is at most $f(a) \le O(RQ^2)+\left\lceil\frac{3rQ^2}{R}\right\rceil \cdot f(a-1)$. Solving this recursion indeed gives $f(a) \le O\left(a \left\lceil\frac{3rQ^2}{R}\right\rceil^a \cdot RQ^2\right)$, as desired. %This is at most $O(RQ^2)+\left\lceil\frac{3rQ^2}{R}\right\rceil \cdot O\left( \left\lceil\frac{3rQ^2}{R}\right\rceil^{a-1}\cdot RQ^2 \right) \leq O\left(\left\lceil\frac{3rQ^2}{R}\right\rceil^a \cdot RQ^2\right)$.\mscomment{prove this by recursion instead of induction. i outlined it in messenger}
\end{proof}

\begin{cor} \label{cor:linear-space}
    With high probability, the amount of memory used by Algorithm~\ref{alg:linear} is at most $O\left(\left\lceil\frac{3rQ^2}{R}\right\rceil^d \cdot RQ^3\right)$. (We will show later that this is at most $s(n)$.)
\end{cor}

\begin{proof}
    The space required to instantiate the algorithm including the lists $\qlist_i$ is less than $O\left(d\left\lceil\frac{3rQ^2}{R}\right\rceil^d \cdot RQ^2\right)$, and so the total space is dominated by the recursive call $\mathsf{recurse\_linear}(d,z,())$, which requires space $O\left(d\left\lceil\frac{3rQ^2}{R}\right\rceil^d \cdot RQ^2\right) \le O\left(\left\lceil\frac{3rQ^2}{R}\right\rceil^d \cdot RQ^3\right)$ by Lemma~\ref{lem:linear-space}. 
\end{proof}

Next, we show the correctness guarantees of the protocol -- that the decoding algorithm outputs the correct answer with high probability. For simplicity of the argument, it'll be convenient to show a lemma about the following function $\mathsf{recurse\_linear}(a,z',T)$. This function is the same as $\mathsf{recurse\_linear}(a,z,T)$, except that the phantom stream $z'\in (\bbK \cup \bot)^{[R]^a}$ instead of in binary. As such, it is the identical procedure, but we leave out Line~\ref{line:C-inner} and instead let the next symbol of $z'$ be $\sigma$ directly.

\begin{lemma} \label{lem:linear-correct}
    Consider a call to the function $\mathsf{recurse\_linear'}(a,z',T)$. 
    %Let $m\in (\bbK \cup \bot)^{r^a}$ be an arbitrary message indexed by tuples in $[r]^a$, and define $\theta:=\delta(z,\LDC^{\otimes a}(m))$. 
    Let $m\in \bbK^{[r]^a}$ be an arbitrary message, let $z'\in (\bbK \cup \bot)^{[R]^a}$, and let $\theta:=\delta(z',\LDC^{\otimes a}(m))$. Then, let $p_{correct}$ be the probability $\mathsf{recurse\_linear}(a,z',T)$ outputs $\ell_T\cdot m$, and $p_\bot$ be the probability it outputs $\bot$. Then, $p_{correct}+0.5p_\bot \geq 1-\theta-2a\eps'$.
\end{lemma}

\begin{proof}
    We'll prove this statement by induction on $a$. 
    %We remark that the indices of stream $z$ being received in $\mathsf{recurse\_linear}(a, z, T)$ are always  (this can be shown easily by induction).
    
    For the base case of $a=0$, since $\LDC^{\otimes 0}(m)=m$, we want to show that $p_{correct}+0.5p_\bot>1-\delta(z',m)$. If $z'=m$, then $p_{correct}=1$, $p_\bot=0$ and $\delta(z',m)=0$ so the equation is satisfied. If $z'=\bot$, then $p_{correct}=0$, $p_\bot =1$ and $\delta(z',m)=0.5$ so the equation is satisfied. If $z'\neq m$, then $p_{correct}\geq 0$, $p_\bot \geq 0$ and $\delta(z',m)=1$, so the equation is satisfied.

    Next, we do the inductive step. The first step of the algorithm is to compute $\cM_i[I] \gets \mathsf{recurse\_linear'}(a-1, z'_I, T|i)$ for all pairs $i\in [r],I\in [R]$ such that $I\in \qlist_i$. The second step is to perform a local decoding for each $i$ using these values. We will analyze this.
    
    Let $\LDC(m)$ be the string whose indices are in $[R]\times[r]^{a-1}$ obtained by applying the encoding function $\LDC$ to $m_1\ldots m_r$ on each bit of the strings separately. For all $I$, define $\theta_I = \delta(z'_I,\LDC^{\otimes a}(m)_I)$.  By the inductive hypothesis, the probability that $\mathsf{recurse\_linear'}(a-1, z'_I, T|i)$ does not output $\ell_{T|i}\cdot \LDC(m)_I$ (giving half credit to outputting $\bot$) is at most $\theta_I+2(a-1)\eps'$. We remark that we do not actually compute $\mathsf{recurse\_linear'}$ for all values of $I$, but we may imagine that we do for the purpose of analysis only.
    
    We'll say that the errors are the $I$ such that $\mathsf{recurse\_linear'}(a-1, z'_I, T|i)$ doesn't match $\ell_{T|i}\cdot \LDC(m)_I$, and that it's half an error if the output is $\bot$.
    By Hoeffding's inequality, the fraction of errors over all $I$ is at most $\theta+2(a-1)\eps'+\eps'/2$ with probability at least $1-\exp(-2\eps'^2 R)$. Note that $\eps'^2R>\eps'^2r>(\log n)^2$ for sufficiently large $n$ relative to $\eps^{-1}$, so this probability is at least $1-\eps'/(4r)$.
    
    For a fixed value of $i$, the goal is to estimate $\ell_{T|i}\cdot m_i$. Let's look at all the values (ranging over $I$) of $\mathsf{recurse\_linear'}(a-1, z'_I, T|i)$. 
    Let $\widehat{m^{(i)}}$ be the string $\ell_{T|i}\cdot m_1\ldots \ell_{T|i} \cdot m_r$. By linearity, $\ell_{T|i}\cdot \LDC(m)_I=\LDC(\widehat{m^{(i)}})_I$. The algorithm's steps are to query the bits in $\qlist_i$ on $\LDC(\widehat{m^{(i)}})$ and perform smooth local decoding on them. By the guarantee of the smooth local decoding algorithm in Theorem~\ref{thm:ldc-large}, with high probability in $r$ (and therefore probability at least $1-\eps'/(4r)$), we will obtain 
    \[
        p_i(\ell_{T|i}\cdot m_i) + 0.5p_i(\bot)>1-(\theta+2(a-1)\eps'+\eps')-\eps' = 1-\theta-2a\eps'+\eps'/2.
    \]
    %One can translate this to the following guarantees:
    
    % \begin{itemize}
    %     \item If $p<\frac14-\eps'/2$, local decoding of $i$ is correct (outputs $\ell_{T|i}\cdot m_i$), and the  confidence output is at least $1-2p-\frac{a\eps'}{R}$. 
    %     \item If $p\geq \frac14-\eps'/2$, local decoding of $i$ is either correct, or incorrect with confidence at most $1-2p+\eps'+\frac{a\eps'}{R}$. 
    % \end{itemize}
    
    Taking a union bound over $i$, this fails for some $i$ with probability $r(\eps'/(4r)+\eps'/(4r))=\eps'/2$. Note that $\sum_i \ell_{T|i}\cdot m_i = \ell_T\cdot m$. Let $p_{min} = \min_i p_i(\ell_{T|i}\cdot m_i)$. Therefore, we will need to show $p_{correct}+p_\bot>1-\theta-2a\eps'+\eps'/2$, conditioned on the equation holding true for all $i$, to account for the probability of failure.
    \begin{itemize}
        \item If $p_{min}$ is nonzero: Then, $p_{correct}=p_{min}$ and $p_\bot=1-p_{min}$, and indeed we have $p_{min}+0.5(1-p_{min})>1-\theta-2a\eps'+\eps'/2$.
        \item If $p_{min}$ is zero: Let $i$ be such that $p_i(\bot)$ is minimized satisfying $p_i(\ell_{T|i}\cdot m_i)=0$. Then $p_{correct}\geq 0$. Moreover, $p_\bot = \max_j p_j(\bot) \geq p_i(\bot)$. Indeed, $0.5p_i(\bot)>1-\theta-2a\eps'+\eps/2$ is satisfied.
    \end{itemize}
    % \begin{itemize}
    %     \item If $p<\frac14-\eps'/2$, the computed return value of the algorithm is in fact $\ell_T\cdot m$, which it returns with probability at least $1-2p-\frac{a\eps'}{R}$.
    %     \item If $p\geq \frac14-\eps'/2$, either all the computed values are correct, in which case the correct answer is returned with probability at least $\frac12$ since confidences are always at least $\frac12$, or at least one of the confidences at most $1-2p+\eps'+\frac{a\eps'}{R}$, so the correct answer is still returned with probability at least $2p-\eps'-\frac{a\eps'}{R}$. \vgnote{Need $2p-\eps' \ge 1-2p-\eps'$, so some adjustment of $\eps'$ to $eps/2$ might be needed.}
    % \end{itemize}
    
    This proves the lemma statement.
\end{proof}

\begin{cor} \label{cor:linear-correct}
    As long as at most $\frac14-\eps'$ of $\enc(x)$ is corrupted, Bob's decoding algorithm $\mathsf{linear\_dec}$ outputs $\ell\cdot x$ correctly with high probability in $n$.
\end{cor}

\begin{proof}
    This follows from $a=d$ in Lemma~\ref{lem:linear-correct}. We remark that that $2d\eps'>\eps/4$, so $p_{correct}+0.5p_\bot\geq 1-\delta(z,\LDC^{\otimes a}(x))-\eps$. Let $z'$ be the string that is decoded when executing $\mathsf{recurse\_linear}$ via Line~\ref{line:C-inner}. There is only one such value ever computed, since all the parallel executions execute the computation of $z'$ simultaneously. Lemma~\ref{lem:linear-correct} says that $p_{correct}+0.5p_\bot\geq 1-\delta(z',\LDC^{\otimes a}(x))-\eps/4$. It remains to understand $\delta(z',\LDC^{\otimes a}(x))$. Let us view this symbol by symbol, yielding each symbol of $z'$ from decoding $\C_{inner}$. For any index $a$, we observe that $\bbE\left[ \delta(z'_a,\LDC^{\otimes a}(x)_a) \right] \leq (2+\eps)\cdot\text{ corruption fraction}$. Recalling the total corruption fraction is $\frac14-\eps$, the total relative distance with high probability is at most $(2+\eps)(1/4-\eps)\leq 1/2-\eps$.

    Now that we have established that $\delta(z',\LDC^{\otimes a}(x))\leq 1/2-\eps$ with high probability, it holds that 
    \[
        p_{correct}+0.5p_\bot\geq 1-\delta(z,\LDC^{\otimes a}(x))-\eps/4 \geq 1/2+\eps/2.
    \]

    The probability of outputting the correct answer is at least $p_{correct}+0.5p_\bot-p_{failure}$, so by Chernoff bound the probability that $\ell\cdot x$ is output is at least $\frac12+\eps/4$ of the time is $1-\frac{1}{n^{\omega(1)}}$ as desired.
\end{proof}

\begin{proof}[Proof of Theorem~\ref{thm:linear-const}]

We return to the main proof of Theorem~\ref{thm:linear-const}. Corollary~\ref{cor:linear-correct} tells us that Bob's decoding algorithm yields the correct output with high probability. Corollary~\ref{cor:linear-space} tells us that the required space is $O\left(\left\lceil\frac{3rQ^2}{R}\right\rceil^d \cdot RQ^3\right)$. We now choose parameters and evaluate this in the three contexts of the theorem statement. We note that $Q$ and $R$ are the only parameters that can be chosen by us, subject to $Q>\left( d \log_Q r\right)^{1000}$ and $R\geq (d \log_Q r)^{150\log_Q r}$ and $Q<r$. 

\paragraph{When $s(n)= (\log n)^t$:}
    Then, $r(n) = (\log n)^{0.2t}$. Choose $Q=(\log n)^{0.01\sqrt{t}}$ and $R=r\cdot(\log n)^{200\log_Q r}$. Note that $\log_Q r=20\sqrt{t}$.
    %Also, note that $6C+1<\log_Q r < t/(15000)$ for sufficiently large $n$. 
    Then, the bound on space is
    \begin{align*}
        O\left(\left\lceil\frac{3rQ^2}{R}\right\rceil^d \cdot RQ^2\right) 
        &\leq O\left(\left\lceil\frac{3r\cdot (\log n)^{0.02\sqrt{t}}}{r\cdot (\log n)^{200\log_Q r}}\right\rceil^d \cdot r \cdot (\log n)^{200\log_Q r} \cdot (\log n)^{0.03\sqrt{t}} \right) \\
        &\leq O\left(\left\lceil\frac{3r\cdot (\log n)^{0.02\sqrt{t}}}{r\cdot (\log n)^{20\sqrt{t}}}\right\rceil^d \cdot r \cdot (\log n)^{4000\sqrt{t}} \cdot (\log n)^{0.03\sqrt{t}} \right) \\
        &\leq O\left( 1\cdot r \cdot (\log n)^{4001\sqrt{t}} \right) \\
        &\leq s(n).
    \end{align*}

    We bound the rate as follows. Note that $d=\log_r n = \log_{\log n} (n^{5/t})$.
    %We will use that $d=\log_Q n \leq \log_{\log n} n^{5/t}$.
    \[
        \left(\frac{r}{R}\right)^d \geq \left(\frac{r}{r\cdot (\log n)^{200\log_Q r}} \right)^d \geq (\log n)^{-20d\sqrt{t}} \geq n^{-100/\sqrt{t}}
    \]
    which is exactly what we desired.

% \paragraph{When $s(n) = (\log n)^{\omega(1)}$:} Let $Q=(\log n)^{0.01\sqrt{\log_{\log n} s(n)}}$ and $R=(\log n)^{200\log_Q r}$. Although this choice of variable appears complicated, it's just a generalization of the previous case. Note that $Q=s(n)^{0.01/\sqrt{\log_{\log n} s(n)}}$ and that $\log_Q r = \log_{\log n} \left( s(n)^{20/\sqrt{\log_{\log n} s(n)}} \right)$ which implies that $R=r\cdot s(n)^{20/\sqrt{\log_{\log n} s(n)}}$.

%     \begin{align*}
%         &O\left(\left\lceil\frac{3rQ^2}{R}\right\rceil^d \cdot RQ^2\right) \\
%         &\leq O\left(\left\lceil\frac{3r\cdot s(n)^{0.03/\sqrt{\log_{\log n} s(n)}}}{r\cdot \left( s(n)^{20/\sqrt{\log_{\log n} s(n)}} \right)}\right\rceil^d \cdot r \cdot \left( s(n)^{20/\sqrt{\log_{\log n} s(n)}} \right) \cdot s(n)^{0.03/\sqrt{\log_{\log n} s(n)}} \right) \\
%         &\leq O\left( 1\cdot r \cdot s(n)^{21/\sqrt{\log_{\log n} s(n)}} \right) \\
%         &\leq s(n).
%     \end{align*}

% The last inequality follows because $\sqrt{\log_{\log n} s(n)}$ is super-constant since $s(n)>\polylog(n)$.

% We bound the rate as follows. $d=\log_r n = \log_{s(n)} n^5$.
% %Note that $d=\log_r n = \log_{\log n} (n^{5/t})$.
%     \[
%         \left(\frac{r}{R}\right)^d \geq \left(\frac{r}{r\cdot s(n)^{20/\sqrt{\log_{\log n} s(n)}}} \right)^d \geq s(n)^{-20\cdot (\log_{s(n)} n^5)/\sqrt{\log_{\log n} s(n)}} \geq n^{-100/\sqrt{\log_{\log n} s(n)}}
%     \]
% which is $n^{-o(1)}$ because $\sqrt{\log_{\log n} s(n)}$ is super-constant.

\paragraph{When $s(n)=n^{\delta}$:} Let $Q=n^{0.01\delta^2}$ and $R=r\cdot 2^{10000/\delta^2}$. Note that $\log_Q r = 20/\delta$ and $d=\log_r n = 5/\delta$. Therefore, $(d \log_Q r)^{150\log_Q r} < (100\delta^{-2})^{3000/\delta} < 2^{10000/\delta^2} = R$. Let us calculate the space (assuming $n$ is sufficiently large):

\begin{align*}
        O\left(\left\lceil\frac{3rQ^2}{R}\right\rceil^d \cdot RQ^2\right) 
        &\leq O\left(\left\lceil\frac{3r\cdot n^{0.02\delta^2}}{r\cdot 2^{10000/\delta^2}}\right\rceil^d \cdot r\cdot 2^{10000/\delta^2} \cdot n^{0.03\delta^2} \right) \\
        &\leq O\left((n^{0.1\delta^2})^d \cdot r\cdot 2^{10000/\delta^2} \cdot n^{0.03\delta^2} \right) \\
        &\leq O\left(n^{0.5\delta} \cdot r\cdot n^{0.1\delta^2} \right) \\
        &\leq s(n).
    \end{align*}

Next, we bound the rate as follows:

    \[
        \left(\frac{r}{R}\right)^d \geq \left(\frac{r}{r\cdot 2^{10000/\delta^2}} \right)^d = \left(2^{10000/\delta^2}\right)^{-5/\delta} \geq 2^{-50000/\delta^3}
    \]
which is in fact constant.
\end{proof}

\bibliographystyle{alpha}
\bibliography{refs}

\end{document}